\pdfoutput=1 

\documentclass[acmsmall]{acmart}

\acmJournal{PACMPL}
\acmVolume{1}
\acmNumber{CONF} 
\acmArticle{1}
\acmYear{2018}
\acmMonth{1}
\acmDOI{} 
\startPage{1}

\setcopyright{none}

\usepackage{booktabs}
\usepackage{subcaption}
\usepackage{amsmath, amsfonts}
\usepackage{tikz}
\usetikzlibrary{arrows,matrix,decorations.pathmorphing,
  decorations.markings, calc, backgrounds}
\usepackage{mathpartir}
\usepackage[nameinlink,capitalize,noabbrev]{cleveref}

\usepackage{breakcites}     
\usepackage[utf8]{inputenc} 

\usepackage{agda/latex/agda}
\newcommand{\anum}[1]{\AgdaNumber{#1}}
\newcommand{\symb}[1]{\AgdaSymbol{#1}}
\newcommand{\data}[1]{{\AgdaDatatype{#1}}}

\newcommand{\module}[1]{{\AgdaModule{#1}}}
\newcommand{\field}[1]{{\AgdaField{\ensuremath{\mathsf{#1}}}}}
\newcommand{\proj}[1]{\;.\field{#1}}
\newcommand{\func}[1]{{\AgdaFunction{#1}}}

\newcommand{\primty}[1]{{\AgdaPrimitiveType{#1}}}

\newcommand{\iz}[0]{\con{i0}}
\newcommand{\io}[0]{\con{i1}}

\newcommand{\var}[1]{{\AgdaBound{#1}}}

\newcommand{\con}[1]{{\AgdaInductiveConstructor{\ensuremath{\mathsf{#1}}}}}
\usepackage{catchfilebetweentags}
\usepackage{newunicodechar}
\newunicodechar{≃}{\ensuremath{\simeq}}
\newunicodechar{≅}{\ensuremath{\cong}}
\newunicodechar{∈}{\ensuremath{\in}}
\newunicodechar{⁻}{\ensuremath{^{-}}}
\newunicodechar{ᴹ}{\ensuremath{_{M}}}
\newunicodechar{ᴺ}{\ensuremath{_{N}}}
\newunicodechar{≡}{\ensuremath{\equiv}}
\newunicodechar{λ}{\ensuremath{\lambda}}
\newunicodechar{⊎}{\ensuremath{\uplus}}
\newunicodechar{∷}{\ensuremath{::}}
\newunicodechar{ℓ}{\ensuremath{\ell}}
\newunicodechar{ᵢ}{\ensuremath{_i}}
\newunicodechar{⟨}{\ensuremath{\langle}}
\newunicodechar{⟩}{\ensuremath{\rangle}}
\newunicodechar{α}{\ensuremath{\alpha}}
\newunicodechar{β}{\ensuremath{\beta}}
\newunicodechar{θ}{\ensuremath{\theta}}
\newunicodechar{φ}{\ensuremath{\varphi}}
\newunicodechar{ψ}{\ensuremath{\psi}}
\newunicodechar{η}{\ensuremath{\eta}}
\newunicodechar{ε}{\ensuremath{\varepsilon}}
\newunicodechar{ι}{\ensuremath{\iota}}
\newunicodechar{Σ}{\ensuremath{\Sigma}}
\newunicodechar{∀}{\ensuremath{\forall}}
\newunicodechar{ℕ}{\ensuremath{\mathbb{N}}}
\newunicodechar{→}{\ensuremath{\to}}
\newunicodechar{⊎}{\ensuremath{\uplus}}
\newunicodechar{⋆}{\ensuremath{*}}
\newunicodechar{¬}{\ensuremath{\lnot}}
\newunicodechar{Θ}{\ensuremath{\theta}}
\newunicodechar{ᴸ}{\ensuremath{{}^{\color{black}\leftarrow}}}  
\newunicodechar{ᴿ}{\ensuremath{{}^{\color{black}\rightarrow}}} 

\ifXeTeX
\else
\newunicodechar{∧}{\ensuremath{\wedge}}
\newunicodechar{∨}{\ensuremath{\vee}}
\newunicodechar{∼}{\ensuremath{\sim}}
\newunicodechar{≢}{\ensuremath{\nequiv}}
\newunicodechar{Π}{\ensuremath{\Pi}}
\newunicodechar{ℤ}{\ensuremath{\mathbb{Z}}}
\newunicodechar{∥}{\ensuremath{\parallel}}
\newunicodechar{∣}{\ensuremath{\mid}}
\newunicodechar{ℚ}{\ensuremath{\mathbb{Q}}}
\newunicodechar{₊}{\ensuremath{_+}}
\newunicodechar{∗}{\ensuremath{\ast}}
\newunicodechar{₁}{\ensuremath{_1}}
\newunicodechar{₂}{\ensuremath{_2}}
\newunicodechar{⊥}{\ensuremath{\bot}}
\newunicodechar{∘}{\ensuremath{\circ}}
\newunicodechar{ρ}{\ensuremath{\rho}}
\fi

\newcommand{\Type}{\func{Type}}
\newcommand{\ptrunc}[1]{\func{∥}\,{#1}\,\func{∥}}
\newcommand{\tyProduct}[2]{{#1}\,\AgdaOperator{\AgdaFunction{×}}\,{#2}}
\newcommand{\tySigma}[3]{\func{Σ[}\,{#1}\,\func{∈}\,{#2}\,\func{]}\,{#3}}
\newcommand{\tyMaybe}[1]{\data{Maybe}\,{#1}}
\newcommand{\tyNat}{\data{ℕ}}
\newcommand{\tyPath}[2]{{#1}\,\func{≡}\,{#2}}
\newcommand{\tyPathP}[4]{\primty{PathP}\,(\symb{λ}\,{#1} \to {#2})\,{#3}\,{#4}}
\newcommand{\tyEquiv}[2]{{#1}\,\func{≃}\,{#2}}

\newcommand{\tyStructEq}[3]{{#1}\,\func{≃[}\,{#2}\,\func{]}\,{#3}}

\newcommand{\tmPair}[2]{{#1}\,\con{,}\,{#2}}
\newcommand{\tmTrunc}[1]{\con{[}\,{#1}\,\con{]}}
\newcommand{\tmNil}{\con{[]}}

\newenvironment{displaymathcode}
  {\par\vspace{\abovedisplayskip}
   \noindent\hspace{\mathindent}\(}
  {\)\par\vspace{\belowdisplayskip}}

\newcommand{\Agda}{{\tt Agda}}
\newcommand{\Coq}{{\tt Coq}}
\newcommand{\Lean}{{\tt Lean}}
\newcommand{\Idris}{{\tt Idris}}
\newcommand{\CubicalAgda}{{\tt Cubical} {\tt Agda}}

\definecolor{Revolutionary}{RGB}{232,70,68}

\newcommand{\substnop}[2]{{#2}\, / \,{#1}}

\definecolor{dkblue}{rgb}{0,0.1,0.5}
\definecolor{lightblue}{rgb}{0,0.5,0.5}
\definecolor{dkgreen}{rgb}{0,0.6,0}
\definecolor{dkbrown}{rgb}{0.4,0,0}
\definecolor{dkviolet}{rgb}{0.6,0,0.8}

\begin{document}

\title{Internalizing Representation Independence with Univalence}




\author{Carlo Angiuli}
\orcid{0000-0002-9590-3303}
\affiliation{
  \department{Computer Science Department}
  \institution{Carnegie Mellon University}
  \streetaddress{5000 Forbes Avenue}
  \city{Pittsburgh}
  \state{PA}
  \postcode{15213}
  \country{USA}
}
\email{cangiuli@cs.cmu.edu}

\author{Evan Cavallo}
\orcid{0000-0001-8174-7496}
\affiliation{
  \department{Computer Science Department}
  \institution{Carnegie Mellon University}
  \streetaddress{5000 Forbes Avenue}
  \city{Pittsburgh}
  \state{PA}
  \postcode{15213}
  \country{USA}
}
\email{ecavallo@cs.cmu.edu}

\author{Anders M\"ortberg}
\affiliation{
  \department{Department of Mathematics}
  \institution{Stockholm University}
  \streetaddress{Kr\"aftriket Hus 6}
  \city{Stockholm}
  \postcode{10691}
  \country{Sweden}
}
\email{anders.mortberg@math.su.se}

\author{Max Zeuner}
\affiliation{
  \department{Department of Mathematics}
  \institution{Stockholm University}
  \streetaddress{Kr\"aftriket Hus 6}
  \city{Stockholm}
  \postcode{10691}
  \country{Sweden}
}
\email{zeuner@math.su.se}

\begin{abstract}
  In their usual form, representation independence metatheorems
  provide an external guarantee that two implementations of an
  abstract interface are interchangeable when they are related by an
  operation-preserving correspondence. If our programming language is
  dependently-typed, however, we would like to appeal to such
  invariance results within the language itself, in order to obtain
  correctness theorems for complex implementations by transferring
  them from simpler, related implementations. Recent work in proof
  assistants has shown that Voevodsky's univalence principle allows
  transferring theorems between isomorphic types, but many instances
  of representation independence in programming involve non-isomorphic
  representations.

  In this paper, we develop techniques for establishing internal relational
  representation independence results in dependent type theory, by using higher
  inductive types to simultaneously quotient two related implementation types by
  a heterogeneous correspondence between them. The correspondence becomes an
  isomorphism between the quotiented types, thereby allowing us to obtain an
  equality of implementations by univalence. We illustrate our techniques by
  considering applications to matrices, queues, and finite multisets. Our
  results are all formalized in \CubicalAgda{}, a recent extension of \Agda{}
  which supports univalence and higher inductive types in a
  computationally well-behaved way.
\end{abstract}



\maketitle


\section{Introduction}
\label{sec:intro}

There are countless representation independence results in the theory of
programming languages, but common among them is the principle that two
implementations of an abstract interface are contextually equivalent provided
one can exhibit an operation-preserving correspondence between them
\citep{Mitchell86}. Early representation independence theorems were developed in
the context of System F, as ``free theorems'' obtained through parametricity
metatheorems \citep{Wadler89,Reynolds83}; researchers have since scaled these
results to realistic module systems such as those found in \texttt{OCaml} \citep{Leroy95}
and \texttt{Standard ML} \citep{Crary17}.

We consider the problem of representation independence for a programming
language with full-spectrum dependent types. In this setting,
\emph{metatheoretic} invariance results leave much to be desired: ideally, we
would like to appeal to these results \emph{within} the language itself to
obtain correctness theorems at deep discounts, by establishing a correspondence
between two implementations and transferring correctness theorems from the
simpler one to the more complex one.

\paragraph{Proof reuse and transfer}

This problem is closely related to one studied extensively in the context of the
\Coq{} proof assistant \citep{Coq}, namely, the ability to transfer proofs
between representations of mathematical structures
\citep{Magaud03,MagaudBertot02,CohenDenesMortberg13,TabareauTanterSozeau18}. A
simple motivating example is that of the unary natural numbers, which are
well-suited to proofs for their natural induction principle, but too
computationally inefficient to dispatch large arithmetic subgoals
\citep{ChyzakMahboubi+14}; binary numbers and machine integers are much more
efficient, but are difficult to reason about.

In the \Coq{} Effective Algebra Library (\texttt{CoqEAL}), \citet{CohenDenesMortberg13}
developed a framework for rewriting proof goals from \emph{proof-oriented} data
structures and algorithms to \emph{computation-oriented} structures and
algorithms. In their framework, users parameterize their development over the
relevant types, prove a refinement relation between the two instantiations, and
then apply a free theorem obtained via a syntactic translation of the Calculus
of Inductive Constructions (CIC) into itself. The final step uses the fact that the
CIC is sufficiently expressive to define its own parametricity model, albeit
externally \citep{KellerLasson12,BernardyJanssonPaterson12,AnandMorrisett17}.
However, \texttt{CoqEAL}'s transfer framework does not support dependently-typed goals.

Recently, \citet{TabareauTanterSozeau18} combined techniques from \texttt{CoqEAL} with
Voevodsky's univalence axiom, which asserts that isomorphic types are equal (and
thence, interchangeable) \citep{Voevodsky10cmu}. Their \emph{univalent
parametricity} approach transfers theorems, including dependently-typed ones,
between isomorphic types without appealing to a meta-level translation. However,
their framework is limited to isomorphisms, which does not cover all cases of
mathematical interest: for example, \texttt{CoqEAL} handles sparse representations of
matrices, which are often not unique, and are therefore not isomorphic to
standard representations.

\paragraph{Isomorphisms are not enough}

Restricting our attention to isomorphisms is particularly untenable in the
context of programming. Consider a standard \func{Queue}\;\var{A} interface,
consisting of a type \field{Q}, an \field{empty} queue, and functions
$\field{enqueue} : A \to \field{Q} \to \field{Q}$ and $\field{dequeue} :
\field{Q} \to \func{Maybe}\,(\field{Q}\,\func{×}\,A)$. If we intend to prove our
code correct, we may reach for a naive implementation
$\field{Q}\;\AgdaSymbol{=}\;\func{List}\;\var{A}$, which \field{enqueue}s on the
front and \field{dequeue}s from the back (writing \func{last} for the function
sending \con{[]} to \con{nothing} and
$\var{xs}\;\func{++}\;\con{[}\;\var{x}\;\con{]}$ to
$\con{just}\;(xs\;\con{,}\;x)$):
\ExecuteMetaData[agda/latex/Intro.tex]{ListQueue}

If we intend to run our code on serious workloads, however, we may reach instead
for batched queues
($\field{Q}\;\AgdaSymbol{=}\;\func{List}\;\var{A}\;\func{×}\;\func{List}\;\var{A}$),
which \field{enqueue} on the front of the first list and \field{dequeue} from
the front of the second list \citep[Section 5.2]{Okasaki99}:
\ExecuteMetaData[agda/latex/Intro.tex]{BatchedQueue}

Whereas $\func{ListQueue}\;A\proj{dequeue}$ takes linear time,
$\func{BatchedQueue}\;A\proj{dequeue}$ is amortized constant-time (being
constant-time except when \var{ys} is empty, in which case one must
\func{reverse}\;\var{xs}). To understand the batched representation, we can
observe that each \func{BatchedQueue} has the same extensional behavior as
exactly one \func{ListQueue}, computed as follows:
\ExecuteMetaData[agda/latex/Intro.tex]{appendReverse}

This correspondence is \emph{structure-preserving}---it preserves \field{empty}
and commutes with \field{enqueue} and \field{dequeue}---and therefore, by
representation independence, \func{ListQueue}s and \func{BatchedQueue}s are \\
contextually equivalent. We are tantalizingly close to obtaining the same result
internally by univa- \\ lence, or rather, by the \emph{structure identity
principle}~(SIP), a consequence of univalence which asserts that isomorphic
\emph{structured} types (types equipped with operations) are equal, provided
that the iso- \\ morphism is structure-preserving \cite[Section 9.8]{HoTT13}. The gap
is simple: \func{appendReverse} \emph{is not an isomorphism}, as it sends both
$(\con{[]}\,\con{,}\,\con{[}1\,\con{,}\,0\con{]})$ and
$(\con{[}0\con{]}\,\con{,}\,\con{[}1\con{]})$ to $\con{[}0\,\con{,}\,1\con{]}$.

We arrive at the core of our problem: representation independence metatheorems
are unsatisfactory, because we want to leverage representation independence to
prove internal correctness theorems; \texttt{CoqEAL}-style approaches do not allow us to
replace \func{ListQueue}s with \func{BatchedQueue}s in dependently-typed goals;
and univalence does not apply, because our types are not isomorphic.

\paragraph{Contribution}

In this paper, we develop techniques for establishing \emph{internal} relational
representation independence results in dependent type theory, without appealing
to parametricity translations. Taking inspiration from Homotopy Type
Theory/Univalent Foundations (HoTT/UF), we use \emph{higher inductive types}
(HITs) to simultaneously quotient two related implementation types by a
heterogeneous correspondence between them, improving that correspondence to an
isomorphism; we then apply univalence (and in particular, the SIP) to obtain an
equality of implementations. In the case of \func{Queue}s, our technique leaves
\func{ListQueue}s untouched while equating two \func{BatchedQueue}s if and only if
\func{appendReverse} maps them to the same list.

Our results are all formalized in \CubicalAgda{} \citep{VezzosiMortbergAbel19},
an extension of \Agda{} \citep{Agda} which supports univalence and HITs in a
computationally well-behaved way using ideas from Cubical Type Theory
\citep{CCHM18,AngiuliFavoniaHarper18}.

From the perspective of proof reuse, our work bridges the gap between \texttt{CoqEAL} and
univalent parametricity, because it applies to both dependently-typed goals and
non-isomorphism correspondences. In addition, our techniques require less
engineering work than these systems, because we rely on built-in features of
\CubicalAgda{} rather than typeclass instances for univalence or a \Coq{} plugin
to generate parametricity translations. However, we stress that our theorems
hold in any dependent type theory with univalence and HITs, features which can
be added axiomatically to \Coq{} \citep{BauerGross+17} and \Lean{}
\citep{VanDoornVonRaumerBuchholtz17}, among other languages.

\paragraph{Outline}

In \cref{sec:cubicalagda}, we review some features of \CubicalAgda{} that we use
in the remainder of the paper, including univalence and HITs.
In \cref{sec:sip}, we state and prove a version of the SIP,
which states that structure-preserving isomorphisms yield equal
structured types; we also develop reflection-based automation for automatically
deriving the appropriate notion of structure-preservation for a given structure.
In \cref{sec:examples}, we demonstrate how to use the SIP in concert with HITs
to obtain representation independence results for matrices and queues.
In \cref{sec:relational}, we generalize these techniques by introducing the
notion of a \emph{quasi--equivalence relation} (QER), and present the main
technical result of the paper: any structured QER can be improved to a
structured isomorphism between quotients. To illustrate the technique, we apply
our theorem to two implementations of finite multisets.
Finally, in \cref{sec:related}, we conclude with future directions and a
discussion of related work.

\section{Programming in Cubical Type Theory}
\label{sec:cubicalagda}

In this section, we review some basic features of \CubicalAgda{}, a recent
extension of \Agda{}. All code in this paper is formalized in \Agda{} 2.6.2,
which is in development at the time of writing, and our main results have been
integrated into the \texttt{agda/cubical} library.%
  \footnote{The current development version of \Agda{} is available at
  \url{https://github.com/agda/agda}, the \texttt{agda/cubical} library is
  available at \url{https://github.com/agda/cubical}, and our formalization is
  available at \url{https://github.com/agda/cubical/blob/master/Cubical/Papers/RepresentationIndependence.agda}.}
Experts on Cubical Type Theory or HoTT/UF can safely skim this section; readers
who wish to learn more about \CubicalAgda{} are encouraged to consult
\citet{VezzosiMortbergAbel19} for a description of the system, or \citet{CCHM18}
and \citet{CoquandHuberMortberg18} for a detailed exposition of its core type
theory.

\subsection{Equalities as paths}

Most languages with full-spectrum dependent types, including \Agda{} \citep{Agda}, \Coq{}
\citep{Coq}, \Idris{} \citep{Brady13}, and \Lean{} \citep{DeMouraKong+15}, have
two notions of equality. \emph{Definitional equality} is a purely syntactic
notion implemented in typechecking algorithms, which use computation to silently
discharge many trivial obligations. Equations that cannot be discharged in this
way are instead mediated by an \emph{equality type} (or \emph{propositional}
equality) whose proofs are often provided by the user, and appeals to which are
marked explicitly in terms and types.

In those languages, the equality type is defined as an inductive family
$\func{Eq}~A~x~y$ generated by a \emph{reflexivity} constructor $\con{refl}~x :
\func{Eq}~A~x~x$, following \citet{MartinLof75itt}. One can easily show that
equality types satisfy many mathematical properties of equality (e.g., symmetry,
transitivity, congruence) but they famously lack several others, including
\emph{function extensionality}: the property that pointwise-equal functions are
equal as functions \citep{BoulierPedrotTabareau17}.

Principles such as function extensionality, univalence \citep{Voevodsky10cmu},
and excluded middle can be added to type theory as axioms. However, axioms in
type theory lack computational content, causing proofs using axioms to become
``stuck,'' and thus less amenable to definitional equality. More importantly,
type theories with stuck terms cannot serve as programming languages because
they lack \emph{canonicity}, the property that all closed terms of natural
number type compute numerals.

Building on ideas from HoTT/UF, Cubical Type Theory is a computationally
well-behaved extension to dependent type theory whose formulation of equality
types, called \emph{path types}, resolve many longstanding issues with \func{Eq}
types, including the failure of function extensionality, propositional
extensionality (\cref{ssec:univalence}), and the inability to define quotients
(\cref{ssec:hits}). In Cubical Type Theory, paths are defined as maps out of an
\emph{interval} type \func{I} which contains two elements $\iz{} : \func{I}$ and
$\io{} : \func{I}$ that are computationally distinct but logically equivalent,
in the sense that no function can map them to unequal objects. Therefore,
functions $\func{f}:\func{I}\to A$ serve as evidence that $\func{f}(\iz{})$ and
$\func{f}(\io{})$ are equal in $A$, just as paths in topological spaces can be
represented by continuous functions out of the real unit interval $[0,1] \subset
\mathbb{R}$. Iterated equality proofs in $A$ are thus functions $\func{I}^n\to
A$ which correspond topologically to squares, cubes, etc. in $A$, leading to the
term \emph{cubical}.

In \CubicalAgda{}, the interval \func{I} is a primitive type which, in addition
to the elements $\iz{} : \func{I}$ and $\io{} : \func{I}$, is equipped with
three operations---\emph{minimum} (\func{\_∧\_} : \func{I} $\to$
\func{I} $\to$ \func{I}), \emph{maximum} (\func{\_∨\_} : \func{I} $\to$ \func{I}
$\to$ \func{I}), and \emph{reversal} (\func{∼{}\_} : \func{I} $\to$
\func{I})---which satisfy the laws of a \emph{De Morgan algebra}, i.e., a
bounded distributive lattice $(\con{i0},\con{i1},\func{\_∧\_},\func{\_∨\_})$
with a De Morgan involution \func{∼{}\_}. Other formulations of Cubical Type
Theory requiring less structure on \func{I} are also
possible~\citep{AngiuliFavoniaHarper18,FootballHockeyLeague19}.

\paragraph{Paths}

Path types are a special form of dependent function type $(i : \func{I})\to
A\;i$ that specify the behavior of their elements on \con{i0} and \con{i1}:
\ExecuteMetaData[agda/latex/Section2.tex]{PathP}
Here, $\var{A} : \func{I} \to \func{Type}~\var{ℓ}$ is a function from \func{I}
to one of \Agda{}'s universes of types. (\Agda{}'s universes are typically
called \func{Set}, but this paper and the \texttt{agda/cubical} library use
\func{Type} instead because, as we will discuss in \cref{ssec:univalence}, the
term \emph{set} has a technical meaning in HoTT/UF.) Elements of \func{PathP}
are $\lambda$-abstractions
\(
\symb{λ} \var{i} \to t : \func{PathP}\;\var{A}\;t[\substnop {i}
  {\con{i0}}]\;t[\substnop {i} {\con{i1}}]
\)
where $t : A\;i$ for $i: \func{I}$. We can apply a path
$\var{p} : \func{PathP}\;\var{A}\;\var{a}_0\;\var{a}_1$ to $\var{r} : \func{I}$,
obtaining $\var{p}\;\var{r} : \var{A}\;\var{r}$, and paths satisfy $\beta$- and
$\eta$-rules definitionally, just like ordinary \Agda{} functions. The only
difference between $\var{p} : \func{PathP}\;\var{A}\;\var{a}_0\;\var{a}_1$ and a
function $f : (i : \func{I})\to A\;i$ is that $p$ is subject to additional
definitional equalities $\var{p}\;\iz{} = \var{a}_0$ and $\var{p}\;\io{} =
\var{a}_1$.

Because the two ``endpoints'' of a path $\var{a}_0 : \var{A}\;\iz{}$ and
$\var{a}_1 : \var{A}\;\io{}$ have different types, \func{PathP} types in fact
represent \emph{heterogeneous} equalities, or \emph{\func{Path}s over
\func{P}aths} in the terminology of HoTT/UF~\cite[Sect. 6.2]{HoTT13}. We can
recover homogeneous (non-dependent) paths in terms of \func{PathP} as follows:
\ExecuteMetaData[agda/latex/Section2.tex]{Path} %
Here, $\{\var{A}\;\AgdaSymbol{=}\;A\}$ tells \Agda{} to bind the implicit
argument \var{A} (the first \var{A}) to a variable (the second $A$) for use on
the right-hand side. Path types allow us to manipulate equality proofs using
standard functional programming idioms. For instance, a constant path represents
a reflexive equality proof.
\ExecuteMetaData[agda/latex/Section2.tex]{refl}

(Note that \func{refl} is a function, rather than a constructor.) We can
directly apply a function to a path in order to prove that dependent functions
respect equality:
\ExecuteMetaData[agda/latex/Section2.tex]{cong}

We can in fact define \func{cong} for $B : A \to \func{Type}~\ell$, but to
increase readability in the remainder of the paper, we henceforth suppress most
universe levels and arguments such as $\{ A : \func{Type} \}$ that are easily
inferred by the reader. As paths are just functions it is now trivial to
\emph{prove} function extensionality, the property that pointwise equal
functions are equal:
\ExecuteMetaData[agda/latex/Section2.tex]{funext}
The proofs of function extensionality for dependent and $n$-ary functions are
equally direct. Since $\func{funExt}$ is \emph{definable} in \CubicalAgda{}, it
has computational content: namely, to swap the arguments to \var{p}. Functional
programmers might recognize this as a special case of the \texttt{flip}
function.

\paragraph{Transport and composition}

One of the key operations of equality in type theory is \emph{transport}, which
sends equalities between types to coercions between those types. In
\CubicalAgda{}, this principle is a special instance of a primitive called
\func{transp}:
\ExecuteMetaData[agda/latex/Section2.tex]{transport}

Consequences of transport include the substitution principle and, using the
minimum operation 
on the interval, the usual induction principle
for inductively-defined \func{Eq} types:
\ExecuteMetaData[agda/latex/Section2.tex]{subst}

\vspace{-\abovedisplayskip}
\ExecuteMetaData[agda/latex/Section2.tex]{J}

Using \func{J} we can easily reproduce the standard type-theoretic proofs that
\func{≡} is symmetric and transitive. One can also prove these properties
directly using \CubicalAgda{}'s primitive \emph{homogeneous composition}
operation \func{hcomp}, which expresses these ``groupoid laws'' and their
higher-dimensional analogues. Although \func{hcomp} is used extensively in the
\texttt{agda/cubical} library, we do not need to use it directly in this paper.

\subsection{Univalence}
\label{ssec:univalence}


The core innovation of HoTT/UF, and the main motivation behind Cubical Type
Theory, is Voevodsky's \emph{univalence} principle
\citep{Voevodsky10cmu,HoTT13}, which states that any equivalence of types
(written \func{≃}, a higher-dimensional analogue of isomorphism) yields an
equality of types, such that \func{transport}ing along the equality applies the
equivalence \citep{Licata16}:
\ExecuteMetaData[agda/latex/Section2.tex]{ua}

\vspace{-\abovedisplayskip}
\ExecuteMetaData[agda/latex/Section2.tex]{uabeta}

Here, $\func{equivFun} : A\;\func{≃}\;B\to A\to B$ returns the function
underlying an equivalence. Because \func{transport} always produces an
equivalence, univalence essentially states that $A\;\func{≃}\;B$ is not only a
necessary but also a sufficient condition for the existence of an equation
$A\;\func{≡}\;B$.

In the past decade, several libraries of formalized mathematics have been
developed around an axiomatic formulation of univalence, including the UniMath
library in \Coq{} \citep{UniMath}, and the HoTT libraries in \Coq{}
\citep{BauerGross+17}, \Agda{} \citep{AgdaHoTT}, and \Lean{}
\citep{VanDoornVonRaumerBuchholtz17}. In contrast, \func{ua} and \func{uaβ} are
definable in \CubicalAgda{} using the \func{Glue} types of \citet{CCHM18}, and
have computational content. In particular, many instances of \func{uaβ} hold
directly by computation (writing \func{not≃} for a proof that $\func{not} :
\func{Bool}\to\func{Bool}$ is an equivalence):
\ExecuteMetaData[agda/latex/Section2.tex]{transportuanot}

Recalling that all constructions in type theory respect the equality type by
fiat, the force of univalence is therefore that \emph{all constructions respect
equivalence/isomorphism}. For instance, without inspecting the definition of
$\func{IsMonoid} : \func{Type} \to \func{Type}$, we can conclude that an
equivalence $A\;\func{≃}\;B$ gives rise to a path $A\;\func{≡}\;B$, which in turn gives
rise to a path $\func{IsMonoid}~A~\func{≡}~ \func{IsMonoid}~B$. In
\cref{sec:sip}, we will use the structure identity principle to show that a path
between two monoids is exactly a monoid isomorphism in the usual sense.

\paragraph{Sets and propositions}

Univalence refutes \emph{uniqueness of identity proofs} (UIP), or Streicher's
axiom K \citep{Streicher93}, because it produces equality proofs in \func{Type}
that are not equal. For example, we can see via \func{transport} that
\func{ua}\;\func{not≃} and \func{refl} are unequal proofs of
\func{Bool}\;\func{≡}\;\func{Bool}:
\ExecuteMetaData[agda/latex/Section2.tex]{notPathnotrefl}

In the presence of univalence, it is important to keep track of which types
satisfy UIP or related principles expressing the complexity of a type's equality
relation. In the terminology of HoTT/UF, a type satisfying UIP is called an
\emph{h-set} (homotopy set, henceforth simply \emph{set}), while a type whose
elements are all equal is called an \emph{h-proposition} (henceforth
\emph{proposition}):
%
%
\ExecuteMetaData[agda/latex/Section2.tex]{isProp}

\vspace{-\abovedisplayskip}
\ExecuteMetaData[agda/latex/Section2.tex]{isSet}

The empty and unit types are propositions, and propositions are closed under
many type formers. For instance, we can show that a dependent function type
valued in propositions is a proposition, using a variation of function
extensionality:
\ExecuteMetaData[agda/latex/Section2.tex]{isPropPi}

Although all proofs of a proposition are equal, propositions cannot in general
be erased---for example, we will soon see that the property of being an
equivalence is a proposition, but it implies the existence of an inverse
function, which has computational content. We will consider a related
programming application in \cref{eg:cost}.

A set is a type whose equality types are propositions, and hence satisfies UIP.
While Cubical Type Theory has many non-sets, including \func{Type} itself, most
of the types used in this paper will be sets. Hedberg's theorem
states that all types with decidable equality are sets~\citep{Hedberg98}, which
includes many datatypes such as \func{Bool}, \func{ℕ}, and \func{ℤ}. In
addition, any proposition is a set, from which it follows that
\func{isProp}\;\var{A} and \func{isSet}\;\var{A} are propositions.

Finally, a type is \emph{contractible} if it has exactly one element:
\ExecuteMetaData[agda/latex/Section2.tex]{isContr}

(Here, we are using the \texttt{agda/cubical} library's notation for \Agda{}'s
built-in \func{Σ}-types, which are \AgdaKeyword{record}s with constructor
\con{\_,\_}, projections \field{fst} and \field{snd}, and definitional $\eta$.)
We can characterize propositions as types whose equality types are contractible,
just as sets are types whose equality types are propositions. Thus contractible
types, propositions, and sets serve as the bottom three layers of an infinite
hierarchy of types introduced by Voevodsky, known as \emph{h-levels}
\citep{Voevodsky10bonn} or \emph{$n$-types} \citep{HoTT13}.

\paragraph{Equivalences and isomorphisms}

An equivalence $\var{A}\;\func{≃}\;\var{B}$ is a function $A\to B$ such that the
preimage of every point in $B$ is contractible:
\ExecuteMetaData[agda/latex/Section2.tex]{equiv}

(In the HoTT/UF literature, \func{preim} is often called the \emph{homotopy
fiber}.) Every equivalence has an inverse which can be constructed by sending
each $y:B$ to the element of its preimage determined by \func{isContr}. In
practice, we often construct equivalences $\var{A}\;\func{≃}\;\var{B}$ by first
constructing ordinary \emph{isomorphisms} \func{Iso}\;\var{A}\;\var{B}
(quadruples of a function $A\to B$, its inverse $B\to A$, and proofs that these
functions cancel in each direction), then applying the lemma
$\func{isoToEquiv} : \func{Iso}\,A\,B \to A\,\func{≃}\,B$.
%

One might wonder, why introduce equivalences at all? The reason is that
$\func{isEquiv}\,\var{f}$ is always a proposition (by \func{isPropΠ} and the
fact that \func{isContr} is a proposition), and thus two elements of
$\var{A}\;\func{≃}\;\var{B}$ are equal if and only if their underlying functions are equal
(because their \func{isEquiv} proofs must agree). In contrast, the property of a
function $f:A\to B$ being an isomorphism (i.e., the triple of a function $B\to
A$ and proofs that it is left and right inverse to $f$) is \emph{not} a
proposition \citep[Theorem 4.1.3]{HoTT13}. A subtle consequence of this fact is
that if one replaces \func{≃} in the statement of univalence with \func{Iso},
the resulting statement is actually \emph{inconsistent} \citep[Exercise
4.6]{HoTT13}! However, for the purposes of this paper, the reader can safely
imagine equivalences as isomorphisms.

%
%

In dependent type theory, equality in \func{Σ}-types is notoriously difficult to
manage, due to the equality of second projections being heterogeneous (usually
stated with transports). In Cubical Type Theory, however, the natural
heterogeneity of \func{PathP} types allows us to straightforwardly characterize
equality in \func{Σ}-types without any transports:
\ExecuteMetaData[agda/latex/Section2.tex]{sigmaeq}

(Here, \con{iso} is the constructor for the \func{Iso} type.) Therefore, we can
always exchange an equality proof in a \func{Σ}-type with a pair of equalities,
which tends to be much easier than manipulating transports. Finally, by
\func{isoToEquiv} it is easy to show that logically equivalent propositions are
equivalent and thus equal by \func{ua}. This principle, \emph{propositional
extensionality}, is another often-assumed axiom in type theory that is provable
(and hence has computational content) in \CubicalAgda{}.
\ExecuteMetaData[agda/latex/Section2.tex]{propext}

\subsection{Higher inductive types}
\label{ssec:hits}

Finally, \CubicalAgda{} natively supports \emph{higher inductive types} (HITs),
a generalization of inductive datatypes which allows for constructors of
equality type
\citep{HoTT13,LumsdaineShulman19,CavalloHarper19,CoquandHuberMortberg18}. In
HoTT/UF, HITs provide analogues of topological spaces such as the circle,
spheres, and torus; in this paper, we use HITs to take quotients of types by
equivalence relations.

\paragraph{Propositional truncation}

Our first example of a HIT is \emph{propositional truncation}, which quotients a
type by the total relation, yielding a proposition. It has two constructors:
\con{|\_|}, which includes elements of \var{A} as elements of \ptrunc{A}, and
\con{squash}, which equates any two elements of \ptrunc{A}.
\ExecuteMetaData[agda/latex/Section2.tex]{proptrunc}

We write functions out of \ptrunc{A} by pattern-matching, noting that
$\con{squash}\;\var{x}\;\var{y}\;\var{i}$ constructs an element of \ptrunc{A}.
For example, we can define the functorial action of \ptrunc{-} on a function as
follows:
\ExecuteMetaData[agda/latex/Section2.tex]{maptrunc}

In addition to checking that these clauses are well-typed, \CubicalAgda{} must
also check that the clause for $\con{squash}\;\var{x}\;\var{y} :
\var{x}\;\func{≡}\;\var{y}$ is a path between $\func{map}\;\var{f}\;\var{x}$ and
$\func{map}\;\var{f}\;\var{y}$, by substituting \iz{} and \io{} for $\var{i} :
\func{I}$ in the left- and right-hand sides of the definition.

We have already seen with \func{isEquiv} that propositions can have
computational content. In the following example, we use propositional truncation
to hide the mathematical content of a type without disturbing its computational
content.

\begin{example}[Cost monad] \label{eg:cost}
By pairing an element of \ptrunc{\func{ℕ}} with the output of a function, we
obtain a counter that can be incremented at will but whose value is ``hidden''
to the equality type:
\ExecuteMetaData[agda/latex/Section2.tex]{cost}

We then define a monad structure on this type which counts the number of binds
as a simple proxy for computation steps, by starting at \anum{0} and
incrementing once at each bind ($\func{\_>>=\_}$). Using \func{Cost≡}, it is
straightforward to prove that these definitions satisfy the monad laws. (In the
definition of $\func{\_>>=\_}$, \func{map2} is a binary version of \func{map}.)
\ExecuteMetaData[agda/latex/Section2.tex]{costmonad}

The computational behavior of this monad is precisely the same as one in which
the propositional truncation is omitted. For example, if we compare a naive
implementation of Fibonacci to a tail-recursive one, the former requires
many more recursive calls, as expected:

\vspace{-\abovedisplayskip}
\begin{center}
\begin{minipage}[t]{1.0\linewidth}
\begin{minipage}[t]{0.5\linewidth}
\ExecuteMetaData[agda/latex/Section2.tex]{fib}
\end{minipage}%
\begin{minipage}[t]{0.5\linewidth}
\ExecuteMetaData[agda/latex/Section2.tex]{fibtail}
\end{minipage}
\end{minipage}
\end{center}

However, because we have truncated the cost, we are able to
prove that \func{fib} and \func{fibTail} are \emph{equal} as functions using
\func{Cost≡}, despite having different runtime behavior:
$\func{fibEq} : \func{fib}\,\func{≡}\,\func{fibTail}$.
%
\end{example}

\paragraph{Set quotients}

The main HIT that we use in this paper is the \emph{set quotient}, which
quotients a type by an arbitrary relation, yielding a set. It has three
constructors: \con{[\_]}, which includes elements of the underlying type,
\con{eq/}, which equates all pairs of related elements, and \con{squash/}, which
ensures that the resulting type is a set:%
  \footnote{If we omitted \con{squash/}, then the unit type quotiented by the
  total relation would be a HIT generated by \con{[ tt ]} and a path from
  \con{[ tt ]} to itself. But this is the circle type \func{S¹}, which is not a
  set \citep[Section 8.1]{HoTT13}.}
\ExecuteMetaData[agda/latex/Section2.tex]{setquot}

Once again, we can write functions out of $A\;\func{/}\;R$ by pattern-matching;
this amounts to writing a function out of $A$ (the clause for \con{[\_]}) which
sends $R$-related elements of $A$ to equal results (the clause for \con{eq/}),
such that the image of the function is a set (the clause for \con{squash/}). We
therefore obtain the universal property of set quotients, which states that
functions from $A~\func{/}~R$ to a set are precisely the functions out of $A$
that respect $R$:
\ExecuteMetaData[agda/latex/Section2.tex]{setquotuniv}

If we additionally assume that \var{R} is a proposition-valued equivalence
relation, then we can show (by propositional extensionality) that the set
quotients are \emph{effective}, in the sense that if
$\con{[}\;a\;\con{]}\;\func{≡}\;\con{[}\;b\;\con{]}$ then $R\;a\;b$
\citep{Voevodsky15unimath}.

\begin{example}[Rational numbers]\label{eg:q}
In \CubicalAgda{}, we can define rational numbers \func{ℚ} (below, left) as
pairs of integers and nonzero natural numbers, set quotiented by equality of
cross multiplication. Contrast this with the standard type-theoretic definition
of rational numbers \func{ℚ'} (below, right) as pairs of \emph{coprime} numbers
\citep[\texttt{rat}]{MathComp}:

\begin{center}
\begin{minipage}[t]{1.0\linewidth}
\begin{minipage}[b]{0.45\linewidth}
\ExecuteMetaData[agda/latex/Section2.tex]{setquotQ}
\end{minipage}%
\begin{minipage}[b]{0.5\linewidth}
\ExecuteMetaData[agda/latex/Section2.tex]{setquotQ'}
\end{minipage}
\end{minipage}
\end{center}
\vspace{-\belowdisplayskip}

In essence, \func{ℚ'} represents rational numbers by choosing a normal-form
representative of each equivalence class; because these normal forms are unique,
one avoids the need for quotients or setoids. However, the tradeoff is that all
operations on \func{ℚ'} must maintain the coprimality invariant---which can be
prohibitively expensive in practice \citep{ChyzakMahboubi+14}---whereas
operations on the set quotient \func{ℚ} are free to return any representative as
their result.
\end{example}

Many quotients, such as the untyped $\lambda$-calculus modulo
$\beta$-conversion, cannot be defined in type theory without HITs, as they lack
a normal form representation. This can be addressed by passing to
\emph{setoids}, which supplant the equality type in favor of an explicitly-given
equivalence relation for each type \citep{BartheCaprettaPons03}. Users of
setoids must manually prove that operations respect these relations, which is
laborious and can lead to subtle bugs: it was recently discovered that the
algebraic hierarchy in \Agda{}'s standard library incorrectly allows a type of
rings to be equipped with one setoid structure used by the addition and a
different setoid structure used by the
multiplication.\footnote{\url{https://lists.chalmers.se/pipermail/agda/2020/012009.html}}

\section{The structure identity principle}
\label{sec:sip}

The structure identity principle (SIP) is the informal principle that properties of mathematical
structures should be invariant under isomorphisms of such structures. Making this precise requires
committing to a notion of structure; consequently, different formalizations of the SIP are possible,
varying in how structures are represented as well as in their scope and generality. The definition
we present is a slight reformulation of one given by \citet{Escardo19}; ours differs by phrasing the
condition in terms of dependent paths, which are particularly convenient in Cubical Type Theory.

\subsection{Structures}

We begin by defining what we mean by a structure; for us, structures are defined over a carrier type
and are equipped with a notion of structure-preserving equivalence. Then, we single out the
structures (with their equivalences) that are well-behaved enough for the SIP to apply, which we
call univalent structures. The basic definitions closely follow Escard\'{o}'s account.

A \emph{structure} is a function $S:\func{Type}\rightarrow\func{Type}$. The type of
\emph{$S$-structures} is defined as follows.
\ExecuteMetaData[agda/latex/Section3.tex]{TypeWithStr}
A notion of $S$-\emph{structure-preserving equivalences} is a term
{$\iota$ : \func{StrEquiv} $S$}, where
\ExecuteMetaData[agda/latex/Section3.tex]{StrEquiv}
Given two $S$-structures $A~B : \func{TypeWithStr}~S$ and an equivalence between the underlying
types $e : \AgdaField{fst}\;A\;\func{≃}\; \AgdaField{fst}\;B$, the type $\iota\; A\; B\; e$ is the
type of proofs that $e$ is an $S$-\emph{structure-preserving equivalence} between $A$ and $B$.  The
type of $S$-structure-preserving equivalences (henceforth, $S$-\emph{structured equivalences})
between $A$ and $B$ is then given by
\ExecuteMetaData[agda/latex/Section3.tex]{Iso}
We say that $(S , \iota)$ defines a \emph{univalent structure}
if we have a term of the following type.
\ExecuteMetaData[agda/latex/Section3.tex]{UnivalentStr}

A univalent structure is a pair $(S,\iota)$ which satisfies the SIP. Our definition is equivalent to
Escard\'{o}'s \emph{standard notion of structure}, but interacts better with cubical machinery.

\begin{theorem}[SIP]
  For $S:\func{Type}\rightarrow\func{Type}$ and $\iota : \func{StrEquiv}\; S$, we have a term
  \ExecuteMetaData[agda/latex/Section3.tex]{SIP}
\end{theorem}
\begin{proof}
  Suppose $\theta:\func{UnivalentStr}\;S\;\iota$, and let $A~B : \func{TypeWithStr}~S$. By applying
  \func{ΣPath≃PathΣ} on the right, we reduce the goal to the following.
  \begin{align*}
    \tyEquiv
    {\tySigma{e}{\tyEquiv{\field{fst}\; A}{\field{fst}\; B}}{(\iota\; A\; B\; e)}}
    {\tySigma{p}{\tyPath{\field{fst}\; A}{\field{fst}\; B}}{(\func{PathP}\; (\symb{λ}\; i\; \symb{→}\; S\; (p\; i))\; (\field{snd}\; A)\; (\field{snd}\; B)})}
  \end{align*}
  The first components are equivalent by univalence, while the input $\theta$ is a proof that the
  second components are equivalent over the first equivalence. It is straightforward to check that
  any pair of equivalences $e_X : \tyEquiv{X}{X'}$ and
  $e_Y : (x : X) \to \tyEquiv{Y\;x}{Y'\;(\func{equivFun}\;e_X\;x)}$ gives rise to an equivalence
  $\tyEquiv{\tySigma{x}{X}{(Y\;x)}}{\tySigma{x'}{X'}{(Y'\;x')}}$; applying this lemma concludes the proof.
\end{proof}

Writing the underlying map of this equivalence explicitly, we have the following.
\ExecuteMetaData[agda/latex/Section3.tex]{sipmap} \qedhere
The \func{SIP} itself is thus a nearly trivial consequence of the definition of univalent
structure. The work, then, is in equipping structures with useful definitions of structured
equivalence that satisfy \func{UnivalentStr}. (Note that there is always a distinctly \emph{useless}
definition of structured equivalence for any structure $S$:
$\iota\;A\;B\;e = \func{PathP}\;(\symb{λ}\; i\; \symb{→}\; S\; (p\; i))\; (\field{snd}\; A)\;
(\field{snd}\; B)$, which is trivially univalent.) In the following, we show how this can be done
systematically and with support by automation.

\paragraph{Axioms}
A structure can typically be separated into two main components: a \emph{raw structure} consisting
of operations on the carrier type, and propositional \emph{axioms} (often in the form of equations)
governing the behavior of the raw structural components. As observed by \citet{Escardo19}, the
latter can be ignored for the purpose of defining structured equivalences.
\begin{definition}
  Let $S:\func{Type}\rightarrow\func{Type}$ be a structure and $\iota$ : \func{StrEquiv} $S$.
  Suppose we have axioms on $S$-structures in the form of a predicate
  $ax:\func{TypeWithStr}\;S\to\func{Type}$. We then define
\ExecuteMetaData[agda/latex/Section3.tex]{Axioms}
\end{definition}

\begin{lemma}\label{AxiomsUnivalentStr}
  Let $S$, $\iota$, and $\var{ax}$ be as above, and assume that \var{ax} is
  proposition-valued. Given some $\theta : \func{UnivalentStr}\;S\;\iota$, we have
  $\func{UnivalentStr}\; (\func{AxiomsStr}\; S\; \var{ax})\; (\func{AxiomsEquivStr}\; \iota\;
  \var{ax})$.
\end{lemma}
\begin{proof}
  Let $(\tmPair{X}{\tmPair{s}{a}})$ and $(\tmPair{Y}{\tmPair{t}{b}}) :
  \func{TypeWithStr}\;(\func{AxiomsStr}\; S\; \var{ax})$ and $e : \tyEquiv{X}{Y}$.
  By $\theta\;e$ and the definition of $\func{AxiomsEquivStr}$,
  $\func{AxiomsEquivStr}\;\iota\;(\tmPair{X}{\tmPair{s}{a}})\;(\tmPair{Y}{\tmPair{t}{b}})$ is
  equivalent to $\tyPathP{i}{S\;(\func{ua}\;e\;i)}{s}{t}$. To see that this is
  equivalent to
  $\tyPathP{i}{\func{AxiomsStr}\;S\;(\func{ua}\;e\;i)}{(\tmPair{s}{a})}{(\tmPair{t}{b})}$,
  note that by $\func{ΣPath≃PathΣ}$, the latter is equivalent to the following
  $\func{Σ}$-type:

  \begin{displaymathcode}
    \tySigma{p}{\tyPathP{i}{S\;(\func{ua}\;e\;i)}{s}{t}}%
    {(\tyPathP{i}{\var{ax}\;(\tmPair{\func{ua}\;e\;i}{p\;i})}{a}{b})}
  \end{displaymathcode}

  Because we have assumed $\var{ax}\;(\tmPair{\func{ua}\;e\;i}{p\;i})$ is a
  proposition, the second component is contractible, and so the $\func{Σ}$-type
  is equivalent to its first component as needed.
\end{proof}

We can use the SIP to transport proofs of axioms between equivalent raw structures to obtain
equivalent structures-with-axioms.

\begin{corollary}[Induced Structures]\label{transportStrAxioms}
  Let $S$, $\iota$ and $\var{ax}$ be as above, and $Θ : \func{UnivalentStr}\;S\;\iota$.
  Suppose we have $(\tmPair{X}{\tmPair{s}{a}}) :
  \func{TypeWithStr}\;(\func{AxiomsStr}\;S\;\var{ax})$
  and $(\tmPair{Y}{t}) : \func{TypeWithStr}\;S$ and a structured equivalence
  $e : \tyStructEq{(\tmPair{X}{s})}{\iota}{(\tmPair{Y}{t})}$. Then there exists
  $b : \var{ax}\;(\tmPair{Y}{t})$ with
  $\tyPath{(\tmPair{X}{\tmPair{s}{a}})}{(\tmPair{Y}{\tmPair{t}{b}})}$.
\end{corollary}
\begin{proof}
  Set $b\;\symb{=}\;\func{subst}\;\var{ax}\;(\func{sip}\;e)\;a$; recall that
  $\func{sip}\;e : \tyPath{(\tmPair{X}{s})}{(\tmPair{Y}{t})}$, so $b$ has type
  $\var{ax}\;(\tmPair{Y}{t})$. It follows by definition that
  $e :
  \tyStructEq{(\tmPair{X}{\tmPair{s}{a}})}{\func{AxiomsEquivStr}\;\iota\;\var{ax}}{(\tmPair{Y}{\tmPair{t}{b}})}$,
    from which we get the desired path by applying \func{sip} once more.
\end{proof}

Before we show how to build raw structures, we first illustrate the utility of the above lemmas.

\begin{example}[Monoids]
We say that $X$ supports a \emph{raw} monoid structure when it has a distinguished neutral element
\var{ε} and a binary operation $\_\cdot\_$:
\ExecuteMetaData[agda/latex/Section3.tex]{RawMonoidStr}
A raw monoid is simply a type equipped with a raw monoid structure:
\ExecuteMetaData[agda/latex/Section3.tex]{RawMonoid}.
A (fully-cooked) monoid structure is a raw monoid structure whose underlying type is a set, and
whose binary operation is unital and associative:
\ExecuteMetaData[agda/latex/Section3.tex]{MonoidStrWithAxioms}
The type of monoids is thus \ExecuteMetaData[agda/latex/Section3.tex]{Monoid}. A structured
equivalence of monoids is a monoid isomorphism in the usual sense: an equivalence that commutes with
the monoid operations.
\ExecuteMetaData[agda/latex/Section3.tex]{MonoidEquiv}
It is now easy to see that the monoid axioms are all proposition-valued: the first axiom requires
the carrier type to be a set, which forces the equations to be propositions. In order to apply the
SIP to monoids, it remains only to check that \func{RawMonoidStructure} is a \func{UnivalentStr}.
Proving this by hand is not so easy, but rather than proving it directly, we will build it up from
elementary combinators in a way that preserves \func{UnivalentStr}.
\end{example}

\subsection{Building structures}
\label{subsec:building}

Having dealt with axioms, we now define a collection of combinators from which we can build raw
structures, their structured equivalences, and proofs they are univalent. We currently support:

\vspace{\abovedisplayskip}
\begin{tabular}{>{\itshape}l<{} >{$}r<{$} >{$}l<{$}}
  Structures & S\;X,T\;X \coloneqq & X \mid A \mid \tyProduct{S\;X}{T\;X} \mid S\;X \to T\;X \mid \tyMaybe{(S\;X)}
\end{tabular}
\vspace{\belowdisplayskip}

The raw fragment of all structures considered in this paper can be built using the above grammar.
(Of course, $\data{Maybe}$ is just one of many inductive types we could consider.)

\paragraph{Constant and pointed structures}

Our two base cases are the constant structure $\symb{λ}\;\_ \to A$ and the pointed structure
$\symb{λ}\;X \to X$. Given $(\tmPair{X}{a})\;(\tmPair{Y}{a'}) : \func{TypeWithStr}\;(\symb{λ}\;\_
\to A)$, an equivalence $e : \tyEquiv{X}{Y}$ is structured when we have a path $\tyPath{a}{a'}$;
this is trivially univalent. For $(\tmPair{X}{x})\; (\tmPair{Y}{y}) :
\func{TypeWithStr}\;(\symb{λ}\;X \to X)$, $e : \tyEquiv{X}{Y}$ is structured when
$\tyPath{\func{equivFun}\;e\;x}{y}$, i.e., the chosen elements are related by the equivalence.
Univalence of the pointed structure is a consequence of $\func{uaβ}$.

\paragraph{Product structures}

We define the product structure of two structures equipped with notions of structured equivalence
$(S_1,\iota_1)$ and $(S_2,\iota_2)$, as follows:
\ExecuteMetaData[agda/latex/Section3.tex]{Products}
The product of \func{UnivalentStr}uctures is univalent:
\ExecuteMetaData[agda/latex/Section3.tex]{productUnivalentStr}
The corresponding HoTT/UF-proof of \citet{Escardo19} uses the type-theoretic Yoneda lemma
\cite[\S2.8]{Rijke12}, but the nice interplay of \func{Σ}-types and dependent path types in
\CubicalAgda{} makes it simple to construct the desired equivalence in the conclusion directly.
\ExecuteMetaData[agda/latex/Section3.tex]{productUnivalentStrDef}
Here, the lemma \func{Σ-cong-equiv} with a non-dependent second argument gives an equivalence of
products from the equivalences on the projections $θ_1\;e$ and $θ_2\;e$.

\paragraph{Function structures}

We say that two $(\lambda\,X \to S\,X \to T\,X)$-structures $f$ and $g$ are related over $e$ when
they take related $S$-structures to related $T$-structures:
\ExecuteMetaData[agda/latex/Section3.tex]{FunctionEquivStr}

\paragraph{Maybe structures}

The definition of structured equivalence for $\symb{λ}\,X → \data{Maybe}\,(S\,X)$ proceeds by cases on
the pair of structures:
\ExecuteMetaData[agda/latex/Section3.tex]{MaybeEquivStr}

\paragraph{Transport structures}

Although we have given a single definition for each type former thus far, we note that
there is often more than one potentially useful definition of structured equivalence for a given
$S : \Type \to \Type$. For example, structured equivalences for $λ\;X → \tyMaybe{X}$ can also be
defined using the functorial action of $\data{Maybe}$ as follows.
\ExecuteMetaData[agda/latex/Section3.tex]{MaybeEquivStr'}
Indeed, this alternative definition is sometimes more convenient; we will use it (in slightly
generalized form) in \cref{subsec:queues}. It belongs to a class of structured equivalence
definitions that arise from a functorial action on equivalences (here, $\func{map-Maybe}$) and which
are captured by the SIP of \citet{CoquandDanielsson13}. We call these \emph{transport structures}
and define them as follows.
\ExecuteMetaData[agda/latex/Section3.tex]{EquivAction}
Rather than a type of proofs that an equivalence is structured, we have an action by the structure
on equivalences; this action is ``correct'' when it agrees with \func{transport}.
(\citeauthor{CoquandDanielsson13} require only that
$\func{equivFun}\;(\alpha\;(\func{idEquiv}\;X))\;s \equiv s$, but univalence implies that this is
equivalent to our definition.) Any such action gives rise to a notion of structured equivalence
defined as follows.
\ExecuteMetaData[agda/latex/Section3.tex]{ActionToStr}
When we have an element of $\func{TransportStr}\;\alpha$, this notion of structured equivalence is
univalent. In fact, $\func{TransportStr}\;\alpha$ and
$\func{UnivalentStr}\;(\func{EquivAction→StrEquiv}\;\alpha)$ are equivalent conditions. Note,
however, that an element of $\func{StrEquiv}\;S$ does not induce an element of
$\func{EquivAction}\;S$ in any useful way; thus our primary definition of structure is the more
permissive of the two, although transport structures are also closed under the grammar at the head
of this section.

Transport structures are particularly convenient when we define structures on function types.  If we
have an action on the domain, however, we can give a more convenient definition as follows.
\ExecuteMetaData[agda/latex/Section3.tex]{FunctionEquivStr+}

\begin{example}[Monoids revisited]\label{eg:monoids-revisited}
  By \cref{AxiomsUnivalentStr}, to prove the SIP for monoids it remains only to show that
  \func{RawMonoidEquiv} is univalent, which we are now ready to do in a systematic way. First,
  observe that pointed and binary operation structures $\lambda\,X\to X$ and $\lambda\,X\to (X\to
  X\to X)$ are univalent with their canonical notions of structured equivalence. This is immediate
  for pointed structures; for binary operations, we obtain the result by applying
  \func{FunctionEquivStr+} twice to the pointed structure. Finally, we
  take the product of these two structures to see that \func{RawMonoidEquiv} defines a univalent
  structure. Putting everything together, we get the desired SIP for monoids:
\ExecuteMetaData[agda/latex/Section3.tex]{MonoidPath}

Let us now illustrate the utility of \cref{transportStrAxioms}. Given $M : \func{Monoid}$ and $N :
\func{RawMonoid}$ and a $\func{RawMonoidEquiv}$ between $N$ and the underlying \func{RawMonoid} of
$M$, \cref{transportStrAxioms} allows us to transport the axioms from $M$ to $N$ so that we obtain
an induced $N' : \func{Monoid}$ with $\tyPath{M}{N'}$. Specialized to the case of unary and binary
numbers, called \func{ℕ} and \func{Bin} in \texttt{agda/cubical} \cite[Section
2.1]{VezzosiMortbergAbel19}, we could show that these types are equal \func{Monoid}s by proving only
that $(\anum{0},+)$ is a monoid on \func{ℕ}, and defining $+_\func{Bin}$ such that the function
underlying the \func{ℕ≃Bin} equivalence is a monoid homomorphism. By definition, \func{ℕ→Bin} sends
\anum{0} to \con{bin0}, so we only have to show that
\begin{displaymathcode}
(x\,y : \func{ℕ}) \to \tyPath{\func{ℕ→Bin}\,(x + y)}{(\func{ℕ→Bin}\,x) +_\func{Bin} (\func{ℕ→Bin}\,y)}
\end{displaymathcode}
\noindent
which can be readily achieved by \func{ℕ}-induction. Therefore, we can
transport the theory of natural numbers from \func{ℕ} to \func{Bin}
without doing any proofs by \func{Bin}-induction. This is
similar to \texttt{CoqEAL}, in which such morphism lemmas were used to
ensure that all proofs could be completed on proof-oriented types while
computation-oriented types were only used for programming, thereby achieving
a clear \emph{separation of concerns} \citep{Dijkstra74}.
\end{example}

\Cref{eg:monoids-revisited} suggests that other algebraic structures can be treated similarly;
indeed, the same proof strategy applies directly to more complex structures like groups and rings.
In fact, these proofs are so uniform that we can automate them.

\paragraph{Automation}

Using \Agda{}'s reflection mechanism, we have defined tactics for automatically generating
definitions of structured equivalence and proofs they are univalent. For
instance, the macro
$\AgdaMacro{AutoEquivStr}\,(\symb{λ}\,(X : \Type) → \tyProduct{X}{(X → X → X)})$ generates a
definition of structured equivalence for raw monoid structures, while
$\AgdaMacro{autoUnivalentStr}\,(\symb{λ}\,(X : \Type) → \tyProduct{X}{(X → X → X)})$ produces a
proof that the definition is univalent. Our tactics opt for transport structures in the domain of a
function type and univalent structures otherwise, but we allow annotations to override the default
choice; for example, the following definition of equivalence for a queue data structure will use the
transport structure for the codomain of the final function.
\ExecuteMetaData[agda/latex/Queue.tex]{RawQueueEquivMacro}
Given such a definition, the macro \AgdaMacro{AutoStructure} removes annotations to produce the
actual structure definition, in this case
$\symb{λ}\,(X : \Type) \to \tyProduct{X}{\tyProduct{(A \to X \to X)}{(X \to
    \tyMaybe{(\tyProduct{X}{A})})}}$.

\section{Representation independence through the SIP}
\label{sec:examples}

Before showing how to generalize the SIP to relational correspondences, we first
show how to use the SIP to obtain representation independence results of
interest to programmers. First, we consider two representations of matrices: as
functions out of a finite set of indices, which are well-suited to proofs,
and as vectors, which are well-suited to computations. Then, we revisit the
queue example from the introduction, and show how to achieve representation
independence through a combination of set quotients and the SIP.

\subsection{Matrices}
\label{subsec:matrices}

Matrices are a standard example in which dependent types can ensure the
well-definedness of operations such as multiplication. We can achieve this in
\Agda{} with length-indexed lists, or vectors:
\ExecuteMetaData[agda/latex/Matrix.tex]{vec}

We can define operations on vectors such as \func{map}, \func{replicate},
addition, and multiplication in the usual way. Unfortunately, this type is
ill-suited for proofs, because one must perform inductions that match the
structure of the functions one has written: it is quite difficult, for example,
to prove that matrix addition \func{addVecMatrix} is commutative. Ideally, one
would prefer to reason about matrix addition in a \emph{pointwise} fashion, in
which commutativity is trivial. We therefore consider a second implementation
better-suited to such reasoning:
\ExecuteMetaData[agda/latex/Matrix.tex]{fin}

If we now fix natural numbers \var{m} and \var{n} and assume \func{G} is an
additive abelian group, it is trivial to define matrix addition and prove it
commutative by function extensionality:
\ExecuteMetaData[agda/latex/Matrix.tex]{addmatrix}

\vspace{-\abovedisplayskip}
\ExecuteMetaData[agda/latex/Matrix.tex]{addmatrixcomm}

Similarly, we can easily define the zero matrix and negation of matrices, and
prove that these form an abelian group. In the absence of function
extensionality, these results are typically stated using setoids, which
complicates the proofs \citep{Wood19}. Alternatively, \texttt{MathComp}
\citep{MathComp} encodes matrices by the graphs of their indexing functions,
which makes proofs simpler without relying on setoids, but requires a lot of
additional theory that is trivialized by using functions directly.

Although \func{FinMatrix} is convenient for proofs, it is less natural for
computations than \func{VecMatrix}. Consider the following integer matrices:

\vspace{-\abovedisplayskip}
\begin{center}
\begin{minipage}[t]{1.0\linewidth}
\begin{minipage}[t]{0.5\linewidth}
  \ExecuteMetaData[agda/latex/Matrix.tex]{exmatrixM}
\end{minipage}%
\begin{minipage}[t]{0.5\linewidth}
  \ExecuteMetaData[agda/latex/Matrix.tex]{exmatrixN}
\end{minipage}
\end{minipage}
\end{center}

Because \func{if} does not compute on variables,
$\tyPath{\func{addFinMatrix}\,\func{M}\,\func{N}}{(λ\,\_\,\_ \to \anum{1})}$ does not hold by
\func{refl}. In a \func{VecMatrix} representation, however, the analogous
equation does follow from computation (i.e., by \func{refl}). Luckily, we can
define functions between the two representations and prove that they form an
equivalence using \func{funExt}, which we can transform into a path by
\func{ua}:
\ExecuteMetaData[agda/latex/Matrix.tex]{eqmatrix}

Using this equivalence, we can easily obtain the aforementioned equation in
\func{FinMatrix} by transferring the analogous equation in \func{VecMatrix}:
\ExecuteMetaData[agda/latex/Matrix.tex]{exmatrix3}

The \func{replaceGoal} lemma takes an equivalence $e$ and transforms a goal of
the form $\tyPath{x}{y}$ into $\tyPath{e^{-1}\,(e\,x)}{e^{-1}\,(e\,y)}$, then
discharges the former with a proof of the latter. Here, the latter proof is
\func{refl} as $\func{addFinMatrix}\,\func{M}\,\func{N}$ and $(λ\,\_\,\_ \to \anum{1})$ are
mapped to definitionally-equal \func{VecMatrix}es.

We could in fact \func{transport} the entire abelian group structure on
\func{FinMatrix} to \func{VecMatrix} along this equivalence, but the operations
obtained this way are very naive: for example, the transported addition on
\func{VecMatrix} converts both arguments to \func{FinMatrix}, adds them with
\func{addFinMatrix}, then converts them back to \func{VecMatrix}. What we
instead want is an abelian group structure on \func{VecMatrix} whose addition is
given by \func{addVecMatrix}.  We can achieve this using the SIP.

Using the tactics discussed in \cref{sec:sip}, we can automatically obtain a
univalent definition of structured equivalence for abelian groups. By
\cref{transportStrAxioms}, we must only show that the function underlying
\func{FinMatrix≃VecMatrix} sends \func{addFinMatrix} to \func{addVecMatrix}. As
\func{FinMatrix} is well-suited for proofs this poses no difficulties, and so we
obtain an induced \emph{and equal} abelian group structure on \func{VecMatrix}
whose addition operation is \func{addVecMatrix}. Using the equality between
abelian groups, we can now transport both proofs and programs between the two
representations.

\subsection{Queues}
\label{subsec:queues}

Thus far, our examples of the SIP have involved equivalent structured types;
now, we turn to our original example of two non-equivalent implementations of
\func{Queue}s. Fix a set $A$. Then a raw queue structure consists of an
\func{empty} queue and \func{enqueue}/\func{dequeue} functions:
\ExecuteMetaData[agda/latex/Queue.tex]{RawQueueStr}

We add some axioms specifying how these operations should behave: that \var{Q}
is a set, that \func{dequeue} of the \func{empty} queue is \con{nothing}, and
what happens if we \func{dequeue} after \func{enqueue}:
\ExecuteMetaData[agda/latex/Queue.tex]{QueueAxioms}


These axioms form a proposition because we have assumed \var{Q} is a set. As in
\cref{subsec:building}, we automate the construction of a univalent raw queue
structure, then add our axioms to obtain a univalent \func{QueueStructure}. The
\func{ListQueue} implementation from \cref{sec:intro} clearly satisfies our
queue axioms and therefore admits a \func{QueueStructure}. Then every
\func{BatchedQueue} (a pair $\func{List}\;A\;\func{×}\;\func{List}\;A$)
corresponds to a single \func{ListQueue} computed by \func{appendReverse}
(the function sending $(\var{xs}\;\con{,}\;\var{ys})$ to
$(\var{xs}\;\func{++}\;\func{reverse}\;\var{ys})$), but this function is not an
isomorphism: multiple \func{BatchedQueue}s are sent to the same
\func{ListQueue}. Worse yet, \func{BatchedQueue} \emph{is not even an instance
  of} \func{QueueStructure}, as it fails to satisfy \func{dequeueEnqueueAxiom}:
we have
$\tyPath{\var{dequeue}\;(\var{enqueue}\;c\;(\tmPair{\con{[}b\,\con{,}\,a\con{]}}{\tmNil}))}{\con{just}\;(\tmPair{(\tmPair{\tmNil}{\con{[}b\,\con{,}\,c\con{]}})}{a})}$
but
$\tyPath{\con{just}\;(\func{returnOrEnq}\;c\;(\var{dequeue}\;(\tmPair{\con{[}b\,\con{,}\,a\con{]}}{\tmNil})))}{\con{just}\;(\tmPair{(\tmPair{\con{[}c\con{]}}{\con{[}b\con{]}})}{a})}$.

We can resolve both of these problems by identifying any two
\func{BatchedQueue}s sent to the same list by \func{appendReverse}, which we can
do either with a set quotient, or equivalently with the HIT:
\ExecuteMetaData[agda/latex/Queue.tex]{2List}

We equip \func{BatchedQueueHIT} with a \func{QueueStructure} as follows. The
\var{empty} queue is \con{Q⟨ [] , [] ⟩}, and we can define \var{enqueue} and
\var{dequeue} functions that respect the \con{tilt} constructor using helper
functions similar to \func{fastcheck}. The \con{tilt} constructor allows us to
shift elements between the ends of the two lists back and forth, thus ensuring
that any two \func{BatchedQueueHIT}s corresponding to the same \func{ListQueue}
are identified. In particular, \func{appendReverse} can then be extended to an
equivalence $\func{BatchedQueueHIT}~\func{≃}~\func{List}\;A$ which is moreover
structure-preserving, giving an induced raw queue structure on
\func{BatchedQueueHIT}. We may now apply the SIP to transfer the axioms
satisfied by \func{ListQueue} to the quotiented \func{BatchedQueue} operations.




\section{Structured equivalences from structured relations}
\label{sec:relational}

\newcommand{\setQuotient}[2]{{#1} \mathbin{\func{/}} {#2}}
\newcommand{\idPropRel}[1]{\func{IdRel}\;{#1}}
\newcommand{\invPropRel}[1]{{#1}^{-1}}
\newcommand{\compPropRel}[2]{{#1} \cdot {#2}}
\newcommand{\graphRel}[1]{\func{graph}\;{#1}}
\newcommand{\equivGraphRel}[1]{\func{equivGraph}\;{#1}}
\newcommand{\eqRelL}[1]{{#1}^{\leftarrow}}
\newcommand{\eqRelR}[1]{{#1}^{\rightarrow}}

In the queue case study above, we improve a structured \emph{relation} between implementations (generated by
\func{appendReverse}) to a structured \emph{equivalence} by considering a quotient of one of the
implementations. In this section, we prove that a generalization of this technique is always applicable as
soon as we have what we call a \emph{quasi--equivalence relation (QER)}.

\subsection{Quasi\texorpdfstring{--}{–}equivalence relations}
\label{subsec:qers}

Because we will use set quotients, it is natural (and ultimately necessary) to restrict our attention to
proposition-valued relations, i.e., $R : X \to Y \to \func{Type}$ such that each $R\;x\;y$ is a proposition.

\begin{lemma}
  \label{lem:rel-operations}
  We have the following operations on prop-valued relations:
  \begin{itemize}
  \item For any prop-valued $R : X \to Y \to \func{Type}$, the
    inverse relation $\invPropRel{R}$ is also prop-valued.
  \item The identity prop-valued relation on $X$ is the truncated path type:
    $\idPropRel{X}\;x_0\;x_1\;\symb{=}\;\ptrunc{\tyPath{x_0}{x_1}}$.
  \item Given two prop-valued relations
    $R : X \to Y \to \func{Type}$ and $S : Y \to Z \to \func{Type}$, we define their prop-valued composite
    $\compPropRel{R}{S}$ using truncation:
    $(\compPropRel{R}{S})\;x\;z\;\symb{=}\;\ptrunc{\tySigma{y}{Y}{\tyProduct{R\;x\;y}{S\;y\;z}}}$.
  \item If $f : X \to Y$ where $Y$ is a set, then its graph,
    $(\graphRel{f})\;x\;y\;\symb{=}\;(\tyPath{f\;x}{y})$, is a prop-valued relation. In particular, if
    $e : \tyEquiv{X}{Y}$, we define $\equivGraphRel{e}\;\symb{=}\;\graphRel{(\func{equivFun}\;e)}$.
  \end{itemize}
\end{lemma}

\begin{definition}
  A prop-valued relation $R : X \to Y \to \func{Type}$ is \emph{zigzag-complete} when for any $r_0 : R\;x\;y$,
  $r_1 : R\;x'\;y$, and $r_2 : R\;x'\;y'$, we have $R\;x\;y'$. This property can be visualized as follows:
  
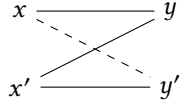
\[
    \begin{tikzpicture}
      \node (x)  at (0,0) { $x$ } ;
      \node (x') at (0,-1) { $x'$ } ;
      \node (y)  at (2,0) { $y$ } ;
      \node (y') at (2,-1) { $y'$ } ;
      \draw [-] (x) to (y) ;
      \draw [-] (x') to (y) ;
      \draw [-] (x') to (y') ;
      \draw [-,dashed] (x) to (y') ;
    \end{tikzpicture}
  \]
\end{definition}

Zigzag-complete relations, also known as \emph{quasi-PERs} or \emph{QPERs} \citep{KrishnaswamiDreyer13}, are a
heterogeneous analogue of partial equivalence relations (PERs). A zigzag-complete relation
$R : X \to Y \to \func{Type}$ expresses a partial correspondence between elements of $X$ and elements of $Y$;
intuitively, if $x$ corresponds to $y$, $x'$ corresponds to $y$, and $x'$ corresponds to $y'$, then $x$ should
also correspond to $y'$. Note that a QPER induces a PER on each of the two types:
$\compPropRel{R}{\invPropRel{R}}$ on $X$ and $\compPropRel{\invPropRel{R}}{R}$ on $Y$. We henceforth denote
these two relations $\eqRelL{R}$ and $\eqRelR{R}$ respectively. In fact, $R$ is zigzag-complete exactly
when $\eqRelL{R}$ and $\eqRelR{R}$ are PERs. In our setting, we require such a partial correspondence to be
\emph{total} by additionally asking for truncated choice functions in either direction.

\begin{definition}
  A zigzag-complete relation $R : X \to Y \to \func{Type}$ is a \emph{quasi--equivalence relation (QER)} when
  there are functions $(x:X)\to\ptrunc{\tySigma{y}{Y}{R\;x\;y}}$ and $(y:Y)\to\ptrunc{\tySigma{x}{X}{R\;x\;y}}$.
\end{definition}

As with QPERs, $R$ is a QER if and only if $\eqRelL{R}$ and $\eqRelR{R}$ are equivalence
relations. Moreover, the truncated choice functions induce an equivalence between the set quotients
$\setQuotient{X}{\eqRelL{R}}$ and $\setQuotient{Y}{\eqRelR{R}}$.

\begin{lemma}
  \label{lem:qer-to-equiv}
  Given a QER $R$, there is an equivalence
  $e : \tyEquiv{\setQuotient{X}{\eqRelL{R}}}{\setQuotient{Y}{\eqRelR{R}}}$ such that for every $x : X$ and
  $y : Y$, we have $\tyPath{\func{equivFun}\;e\;\tmTrunc{x}}{\tmTrunc{y}}$ if and only if $R\;x\;y$.
\end{lemma}
\begin{proof}
  Consider first the forward map, $\setQuotient{X}{\eqRelL{R}} \to \setQuotient{Y}{\eqRelR{R}}$. Observe that
  for every $x : X$, we have \\ a map
  $\con{[\_]}\;\func{∘}\;\field{fst} : (\tySigma{y}{Y}{R\;x\;y}) \to \setQuotient{Y}{\eqRelR{R}}$. Moreover,
  this map is \emph{constant}: given $y\;y' : A$ with \\ $R\;x\;y$ and $R\;x\;y'$, we have
  $\tyPath{\tmTrunc{y}}{\tmTrunc{y'}}$ by definition of the quotient. It therefore factors through the
  propositional truncation \citep{KrausEscardo+17}, giving
  $f_x : \ptrunc{\tySigma{y}{Y}{R\;x\;y}} \to \setQuotient{Y}{\eqRelR{R}}$ such that
  $\tyPath{f_x\;\func{∘}\;\con{[\_]}}{\con{[\_]}\;\func{∘}\;\field{fst}}$. We precompose with the provided
  choice function $(x:X) \to \ptrunc{\tySigma{y}{Y}{R\;x\;y}}$ to get a map $X \to \setQuotient{Y}{\eqRelR{R}}$.
  To show that this map factors through $\setQuotient{X}{\eqRelL{R}}$, it is enough to check that for every
  $x\;x' : X$ with $\eqRelL{R}\;x\;x'$ and $y\;y' : A$ with $R\;x\;y$ and $R\;x'\;y'$, we have
  $\tyPath{\tmTrunc{y}}{\tmTrunc{y'}}$. This follows by composing the relational witnesses and applying
  $\con{eq/}$, using that $R$ is zigzag-complete. Thus we have a map
  $\setQuotient{X}{\eqRelL{R}} \to \setQuotient{Y}{\eqRelR{R}}$; the candidate inverse is constructed in the
  same way.

  The two inverse conditions are propositions, so it suffices to check that they hold on inputs of the form
  $\tmTrunc{x}$ and $\tmTrunc{y}$ respectively, which is a straightforward consequence of the definitions.
  The bi-implication between $\tyPath{\func{equivFun}\;e\;\tmTrunc{x}}{\tmTrunc{y}}$ and $R\;x\;y$ is also a
  straightforward calculation; the forward direction uses effectivity of set quotients by equivalence
  relations \citep{Voevodsky15unimath}.
\end{proof}

\begin{example}[Finite multisets]\label{ex:fmset}
We say that $X$ is a type of \emph{finite multisets (bags) over $A$} if it has an empty multiset $X$, and
functions for insertion $A\to X\to X$, union $X\to X\to X$, and multiplicity $A\to X\to \tyNat{}$. We can
include axioms as well, but for the moment we focus on the raw structure; for simplicity, we also assume that
$A$ has decidable equality. We consider two implementations: lists ($\data{List}\;A$) and \emph{association
lists} ($\func{AssocList}\;A\;\symb{=}\;\data{List}\;(\tyProduct{A}{\tyNat})$), where in the latter case each
element is tagged with a multiplicity. Both implementations take the empty list \con{[]} as the empty multiset; the
remaining operations are defined as follows:

\vspace{-\abovedisplayskip}
\begin{center}
\begin{minipage}[t]{1.0\linewidth}
\begin{minipage}[t]{0.41\linewidth}
  \ExecuteMetaData[agda/latex/Section5.tex]{AddIfEq}

  \vspace{-\abovedisplayskip}
  \ExecuteMetaData[agda/latex/Section5.tex]{ListStructure}
\end{minipage}%
\begin{minipage}[t]{0.59\linewidth}
  \ExecuteMetaData[agda/latex/Section5.tex]{AListStructure}
\end{minipage}
\end{minipage}
\end{center}

Note that \emph{neither} implementation satisfies many of the ``extensional'' laws we expect of multisets,
such as $\tyPath{\func{insert}\;a\;(\func{insert}\;a'\;xs)}{\func{insert}\;a'\;(\func{insert}\;a\;xs)}$. We
now define a QER between $\data{List}\;A$ and $\func{AssocList}\;A$, identifying multisets that assign the
same multiplicity to each element of $A$.
\ExecuteMetaData[agda/latex/Section5.tex]{Relation}

This relation is prop-valued because $\tyNat$ is a set, and is zigzag-complete by transitivity and symmetry of
\func{≡}. We define functions in both directions using the implementations' $\func{insert}$ functions.

\vspace{-\abovedisplayskip}
\begin{center}
\begin{minipage}[t]{1.0\linewidth}
\begin{minipage}[t]{0.4\linewidth}
  \ExecuteMetaData[agda/latex/Section5.tex]{Fwd}
\end{minipage}%
\begin{minipage}[t]{0.6\linewidth}
  \ExecuteMetaData[agda/latex/Section5.tex]{Bwd}
\end{minipage}
\end{minipage}
\end{center}

Finally, proving $\forall\;xs \to \func{R}\;xs\;(\func{φ}\;xs)$ and $\forall\;ys \to
\func{R}\;(\func{ψ}\;ys)\;ys$ takes a bit more work, but the intuition is that these functions preserve the
$\func{count}$ of each element of $A$. Note that the two derived equivalence relations $\eqRelL{\func{R}}$ and
$\eqRelR{\func{R}}$ also wind up identifying multisets on either side precisely when they assign the same
multiplicity to each element of $A$. By \cref{lem:qer-to-equiv}, we have an equivalence between
$\func{List}\;A$ and $\func{AssocList}\;A$ after quotienting by these relations on either side.
\end{example}

\subsection{Structured relations}
\label{subsec:structured-relations}

We now generalize to \emph{structured} relations and equivalences, defining a condition on notions of
structured relation such that any structured QER induces a structured equivalence between quotients. First, as
with structured equivalences, a candidate \emph{notion of structured relation} for a structure $S$ assigns,
for each relation $R$ on types $X$ and $Y$ and structures $s : S\; X$, $t : S\; Y$, a type of witnesses that
the relation is structured:
\ExecuteMetaData[agda/latex/Section5.tex]{StrRel}

(For simplicity, we do not require the input relation to be prop-valued, but our correctness conditions will
only constrain behavior on prop-valued relations.) We now identify a correctness condition on notions of
structured relations.

\begin{definition}
  A candidate $\rho : \func{StrRel}\;S$ is \emph{suitable} when the following hold:
  \begin{enumerate}
  \item \emph{Set- and prop-preservation:} If $X$ is a set, then $S\;X$ is a set. If $R$ is a
    prop-valued-relation, then $\rho\;R$ is a prop-valued relation.
  \item \emph{Symmetry:} For any prop-valued relation $R$, if $(\rho\;R)\;s\;t$, then $(\rho\;(R^{-1}))\;t\;s$.
  \item \emph{Transitivity:} For any prop-valued relations $R,R'$, if $(\rho\;R)\;s\;t$ and
    $(\rho\;R')\;t\;u$, then $(\rho\;\compPropRel{R}{R'})\;s\;u$.
  \item \emph{Descent to quotients:} If $R : X \to X \to \func{Type}$ is a prop-valued equivalence relation
    and $(\rho\;R)\;s\;s$ for some $s : S\;X$, then there is a \emph{unique}
    $\overline{s} : S\;(\setQuotient{X}{R})$ such that we have some
    $r : \rho\;(\graphRel{\con{[\_]}})\;s\;\overline{s}$; that is, the type
    $\tySigma{\overline{s}}{S\;(\setQuotient{X}{R})}{(\rho\;(\graphRel{\con{[\_]}})\;s\;\overline{s})}$ is
    contractible.
  \end{enumerate}
\end{definition}

Note that the condition of being suitable is a proposition. First, we check that these conditions are
sufficient for our original motivation: obtaining structured equivalences from structured QERs.

\begin{theorem}
  \label{thm:qer-descends}
  Let $\rho : \func{StrRel}\;S$ be a suitable notion of structured relation, and let structured types
  $(\tmPair{X}{s}) : \func{TypeWithStr}\;S$ and $(\tmPair{Y}{t}) : \func{TypeWithStr}\;S$ be given. For any
  QER $R : X \to Y \to \func{Type}$ structured by some $r : (\rho\;R)\;s\;t$, the following exist:
  \begin{enumerate}
  \item a structure $\overline{s} : S\; (\setQuotient{X}{\eqRelL{R}})$ with
    $\rho\;(\graphRel{\con{[\_]}})\;s\;\overline{s}$,
  \item a structure $\overline{t} : S\; (\setQuotient{Y}{\eqRelR{R}})$ with
    $\rho\;(\graphRel{\con{[\_]}})\;t\;\overline{t}$,
  \item an element of $\rho\;(\equivGraphRel{e})\;\overline{s}\;\overline{t}$, where
    $e : \tyEquiv{\setQuotient{X}{\eqRelL{R}}}{\setQuotient{Y}{\eqRelR{R}}}$ is the equivalence obtained by
    applying \cref{lem:qer-to-equiv}.
  \end{enumerate}
  In other words, we have the dotted lines in the following picture, where the inner lines indicate relations
  and each outer line indicates an element of $\rho$ over the inner relation.
  
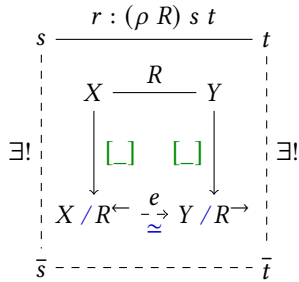
\[
    \begin{tikzpicture}
      \pgfmathsetmacro{\inner}{1.6}
      \pgfmathsetmacro{\outer}{0.7}

      \node (x) at (0,\inner) {$X$} ;
      \node (xrl) at (0,0) {$\setQuotient{X}{\eqRelL{R}}$} ;
      \node (y) at (\inner,\inner) {$Y$} ;
      \node (yrr) at (\inner,0) {$\setQuotient{Y}{\eqRelR{R}}$} ;
      \draw [-] (x) to node [above] {$R$} (y) ;
      \draw [->] (x) to node [right] {\con{[\_]}} (xrl) ;
      \draw [->] (y) to node [left] {\con{[\_]}} (yrr) ;
      \draw [dashed,->] (xrl) to node [above] {$e$} node [below] {\func{≃}} (yrr) ;

      \node (s) at ($(0,\inner) + (-\outer,\outer)$) {$s$} ;
      \node (os) at (-\outer,-\outer) {$\overline{s}$} ;
      \node (t) at ($(\inner,\inner) + (\outer,\outer)$) {$t$} ;
      \node (ot) at ($(\inner,0) + (\outer,-\outer)$) {$\overline{t}$} ;
      \draw [-] (s) to node [above] {$r : (\rho\;R)\;s\;t$} (t) ;
      \draw [dashed,-] (s) to node [left] {$\exists!$} (os) ;
      \draw [dashed,-] (t) to node [right] {$\exists!$} (ot) ;
      \draw [dashed,-] (os) to (ot);
    \end{tikzpicture}
  \]
\end{theorem}
\begin{proof}
  By symmetry and transitivity applied with $r$, we have an element of $(\rho\;\eqRelL{R})\;s\;s$. By descent to
  quotients, we thus obtain $\overline{s} : S\; (\setQuotient{X}{\eqRelL{R}})$ as above. We obtain
  $\overline{t}$ analogously. From $(\rho\;R)\;s\;t$, $\rho\;(\graphRel{\con{[\_]}})\;s\;\overline{s}$,
  and $\rho\;(\graphRel{\con{[\_]}})\;t\;\overline{t}$, we get a proof of
  $\rho\;(\compPropRel{\invPropRel{\graphRel{\con{[\_]}}}}{\compPropRel{R}{\graphRel{\con{[\_]}}}})\;\overline{s}\;\overline{t}$
  by symmetry and transitivity. We arrive at the final condition by checking that the composite relation
  $\compPropRel{\invPropRel{(\graphRel{\con{[\_]}})}}{\compPropRel{R}{\graphRel{\con{[\_]}}}}$
  is equal to $\equivGraphRel{e}$.
\end{proof}

If the restriction of $\rho$ to equivalences is moreover a \emph{univalent} notion of structured equivalence,
we obtain a \emph{path} between $(\tmPair{\setQuotient{X}{\eqRelL{R}}}{\overline{s}})$ and
$(\tmPair{\setQuotient{Y}{\eqRelR{R}}}{\overline{t}})$ as a corollary.

\begin{definition}
  A notion of structured relation $\rho : \func{StrRel}\;S$ is \emph{univalent} when it is suitable and its
  restriction to equivalences, $(\lambda\;(\tmPair{X}{s})\;(\tmPair{Y}{t})\;e \to
  \rho\;(\equivGraphRel{e})\;s\;t) : \func{StrEquiv}\;S$, is univalent.
\end{definition}

As in \cref{sec:sip}, we can show that the collection of suitable and univalent relational structures is
closed under various type formers with reasonable definitions of structured relation. (Notably, the constant
structure $\lambda\;\_ \to A$ is only suitable when $A$ is a set.)

\vspace{\abovedisplayskip}
\begin{tabular}{>{\itshape}l<{} >{$}r<{$} >{$}l<{$}}
  Positive structures & P\;X, Q\;X \coloneqq & X \mid A \mid \tyProduct{P\;X}{Q\;X} \mid \func{Maybe}\;(P\;X)
  \\
  Structures & S\;X,T\;X \coloneqq & P\;X \mid \tyProduct{S\;X}{T\;X} \mid P\;X \to S\;X \mid \func{Maybe}\;(S\;X)
\end{tabular}
\vspace{\belowdisplayskip}

One major departure from the situation in \cref{sec:sip} is that the suitable relation structures are
\emph{not} closed under function types in general, but only in restricted cases. We can see the issue by
considering descent to quotients for the candidate structure $\symb{λ}\; X \to (X \to X) \to X$. Given an
instance of this structure $f : (X \to X) \to X$ and a QER $R : X \to X \to \func{Type}$, we have no way of
obtaining an induced structure $\overline{f} : (\setQuotient{X}{R} \to \setQuotient{X}{R}) \to
\setQuotient{X}{R}$: at some point, we would need to extract maps $X \to X$ from maps $\setQuotient{X}{R} \to
\setQuotient{X}{R}$.  Nevertheless, the function structure $S \to T$ \emph{is} suitable when $S$ belongs to
the restricted class of \emph{positive} relation structures, which we define in two steps as follows.

\begin{definition}
  A candidate $\rho : \func{StrRel}\;S$ \emph{acts on functions} when the following hold:
  \begin{enumerate}
  \item For any $f : X_0 \to X_1$, we have an action $S\; f : S\; X_0 \to S\; X_1$ with
    $\tyPath{S\; (\symb{λ}\,x → x)}{(\symb{λ}\,s → s)}$.
  \item For any $f : X_0 \to X_1$, $g : Y_0 \to Y_1$, relations $R_0 : X_0 \to Y_0 \to \func{Type}$ and
    $R_1 : X_1 \to Y_1 \to \func{Type}$, and
    $\alpha : (x : X_0)\;(y : Y_0) \to R_0\;x\;y \to R_1\;(f\;x)\;(g\;y)$, we have a map
    $\rho\;\alpha : (s : S\;X_0)\;(t : S\;Y_0) \to (\rho\;R_0)\;s\;t \to (\rho\;R_1)\;(S\;f\;s)\;(S\;g\;t)$.
  \end{enumerate}
\end{definition}

\begin{lemma}
  \label{lem:suitable-quotient-map}
  If $\rho : \func{StrRel}\;S$ is suitable and acts on functions, then any prop-valued equivalence relation
  $R : X \to X \to \func{Type}$ induces a map $\setQuotient{(S\;X)}{(\rho\;R)} \to S\;(\setQuotient{X}{R})$.
\end{lemma}
\begin{proof}
  We have $S\;\con{[\_]} : S\;X \to S\;(\setQuotient{X}{R})$. The codomain is a set by set-preservation, so we
  just need to show that $\tyPath{S\;\con{[\_]}\;x_0}{S\;\con{[\_]}\;x_1}$ whenever $(\rho\;R)\;x_0\;x_1$. To
  show this we use the uniqueness condition for descent to quotients: it suffices to show that both
  $S\;\con{[\_]}\;x_0$ and $S\;\con{[\_]}\;x_1$ satisfy $\rho\;(\graphRel{\con{[\_]}})\;x_0$. For the first,
  we start by deriving $(\rho\;\eqRelL{R})\;x_0\;x_0$ from $(\rho\;R)\;x_0\;x_1$ using symmetry and transitivity.
  We have maps $(\symb{λ}\,x → x) : X \to X$ and $\con{[\_]} : X \to \setQuotient{X}{R}$, and know that
  $\eqRelL{R}\;y\;y'$ implies $\graphRel{\con{[\_]}}\;y\;\con{[}\;y'\;\con{]}$ for every $y,y' : X$. The
  action of $\rho$ on functions thus tells us that
  $\rho\;(\graphRel{\con{[\_]}})\;(S\;(\symb{λ}\,x → x)\;x_0)\;(S\;\con{[\_]}\;x_0)$ and therefore
  $\rho\;(\graphRel{\con{[\_]}})\;x_0\;(S\;\con{[\_]}\;x_0)$ holds. The proof of
  $\rho\;(\graphRel{\con{[\_]}})\;x_0\;(S\;\con{[\_]}\;x_1)$ proceeds similarly.
\end{proof}

\begin{definition}
  \label{def:positive}
  Let $\rho : \func{StrRel}\;S$ be a suitable notion of relational structure that acts on functions. We say
  $\rho$ is \emph{positive} when the following hold.
  \begin{enumerate}
  \item \emph{Reflexivity}: For any $X : \func{Type}$ and $s : S\;X$, we have $(\rho\;(\func{IdRel}\;X))\;s\;s$.
  \item \emph{Reverse transitivity}: For any prop-valued relations $R : X \to Y \to \func{Type}$ and
    $R' : Y \to Z \to \func{Type}$ and terms $s : S\; X$ and $t : S\; Z$, the map
    $(\compPropRel{\rho\;R}{\rho\;R'})\;s\;t \to (\rho\;(\compPropRel{R}{R'}))\;s\;t$ supplied by transitivity
    is an equivalence.
  \item \emph{Quotients}: For any prop-valued equivalence relation $R$, the map
    $\setQuotient{(S\;X)}{(\rho\;R)} \to S\;(\setQuotient{X}{R})$ defined in \cref{lem:suitable-quotient-map}
    is an equivalence.
  \end{enumerate}
\end{definition}

These conditions are designed to validate the following result.

\begin{theorem}
  If $\rho_1 : \func{StrRel}\;S_1$ is a positive notion of relational structure and
  $\rho_2 : \func{StrRel}\;S_2$ is a suitable notion of relational structure, then the following is a suitable
  notion of relational structure.
  \ExecuteMetaData[agda/latex/Section5.tex]{FunctionRelStr}
\end{theorem}
\begin{proof}
  The (formalized) proof is rather technical; here, we give just enough detail to point out where the
  conditions in \cref{def:positive} are used.
  Reflexivity and the quotient condition are both used in the proof that \func{FunctionRelStr} induces
  quotient structures. Given an equivalence relation $R$, we must be able to turn maps $f : S_1\;X \to S_2\;X$
  with $\func{FunctionRelStr}\;R\;f\;f$ into maps $S_1\;(\setQuotient{X}{R}) \to S_2\;(\setQuotient{X}{R})$.
  First, we use reflexivity to construct a map $S_1\;X \to S_2\;(\setQuotient{X}{R})$. Given $s : S_1\;X$, we
  have some $r : (\rho_1\;(\func{IdRel}\;X))\;s\;s$ by reflexivity. By the action of $\rho_1$ on functions and
  reflexivity of $R$, this implies $r : (\rho_1\;R)\;s\;s$. Then by definition of
  $\func{FunctionRelStr}\;R\;f\;f$, we have $(\rho_2\;R)\;(f\;s)\;(f\;s)$. By descent to quotients for $\rho_2$,
  we thereby obtain an element of $S_2\;(\setQuotient{X}{R})$.

  The uniqueness condition in descent to quotients allows us to prove that this map
  $S_1\;X \to S_2\;(\setQuotient{X}{R})$ respects $\rho_1\;R$, and thus induces a map
  $\setQuotient{S_1\;X}{\rho_1\;R} \to S_2\;(\setQuotient{X}{R})$. We then apply the quotient condition to
  derive a map $S_1\;(\setQuotient{X}{R}) \to S_2\;(\setQuotient{X}{R})$. Finally, reverse transitivity in the
  domain is naturally required when we prove that $\func{FunctionRelStr}$ is transitive.
\end{proof}

We note that the definitions of suitability and positivity are somewhat negotiable, and are chosen to make
these proofs work. For example, one may instead require suitable (not only positive) structures to be
reflexive. In that case we are still able to prove the closure conditions, but we must also impose a
\emph{reverse} reflexivity condition for positive structures. Our positivity restriction accommodates the
structures currently defined in the \CubicalAgda{} library; the most notable omission (besides $\symb{λ}\; X
\to (X \to X) \to X$, as noted earlier) is $\symb{λ}\; X \to (\tyNat \to X) \to X$.

\begin{example}
To apply \cref{thm:qer-descends} to the multiset implementations from \cref{ex:fmset}, we must show that the QER $\func{R}$ is
structured as a relation between multiset implementations. We generate the notion of structured relation using a
relational equivalent of the tactics described in \cref{sec:sip}. (In this case, an annotation is required on
the constant $\tyNat$ verifying that it is a set.)
\ExecuteMetaData[agda/latex/Section5.tex]{MultiSetStructure}

The module $\module{S}$ defined here packages all the definitions and results we need to work with a relational
structure: $\func{S.structure} : \Type \to \Type$, definitions
$\func{S.equiv} : \func{StrEquiv}\;\func{S.structure}$ and
$\func{S.relation} : \func{StrRel}\;\func{S.structure}$ of structured equivalences and relations respectively,
and proofs that they are univalent. Here, showing that $\func{R}$ is structured requires four conditions:
\ExecuteMetaData[agda/latex/Section5.tex]{areStructured}

The first holds by calculation, while the second holds by definition of $\func{R}$. The third and fourth boil
down to proving that for all $a : A$,
$\tyPath{\func{AL.count}\;a\;(\func{AL.insert}\;x\;ys)}{\func{addIfEq}\;a\;x\;\anum{1}\;(\func{AL.count}\;a\;ys)}$
and
$\tyPath{\func{AL.count}\;a\;(\func{AL.union}\;ys\;ys')}{\func{AL.count}\;a\;ys\;\func{+}\;\func{AL.count}\;a\;ys'}$;
both follow from a case analysis on $ys$ and the outputs of $\func{\_==\_}$. Tupling these four proofs, we
obtain an element of type
$\func{S.relation}\;\func{R}\;(\con{[]}\;\con{,}\;\func{L.insert}\;\con{,}\;\func{L.union}\;\con{,}\;\func{L.count})\;(\con{[]}\;\con{,}\;\func{AL.insert}\;\con{,}\;\func{AL.union}\;\con{,}\;\func{AL.count})$.

It remains only to turn the crank. Applying \cref{thm:qer-descends}, we immediately obtain multiset structures
on $\setQuotient{\func{List}\;A}{\eqRelL{\func{R}}}$ and
$\setQuotient{\func{AssocList}\;A}{\eqRelR{\func{R}}}$ and a path between the two in
$\func{TypeWithStr}\;\func{S.structure}$:

\ExecuteMetaData[agda/latex/Section5.tex]{LQEqualsALQ}

Note that we never need to explicitly show that the multiset operations for $\func{List}\;A$ and
$\func{AssocList}\;A$ preserve $\eqRelL{\func{R}}$ and $\eqRelR{\func{R}}$ respectively; this is a formal
consequence of the proof that $\func{R}$ is structured.

As in \cref{sec:examples}, the path $\func{L/}\eqRelL{\func{R}}\func{≡}\func{AL/}\eqRelR{\func{R}}$ lets us
transfer properties between our quotiented multiset implementations. Consider the axiom that \func{union} is
associative:

\ExecuteMetaData[agda/latex/Section5.tex]{unionAssocAxiom}

This axiom is immediate for $\setQuotient{\func{List}\;A}{\eqRelL{\func{R}}}$ because $\func{L.union}$
(defined as \func{\_++\_}) is associative even prior to quotienting. Then, writing $\func{LUnionAssoc} :
\func{unionAssocAxiom}\; (\setQuotient{\func{List}\;A}{\eqRelL{\func{R}}}\;\con{,}\;\func{LMultisetStr})$,
associativity of \func{union} for $\setQuotient{\func{AssocList}\;A}{\eqRelR{\func{R}}}$ follows by
\func{transport}ing across $\func{L/}\eqRelL{\func{R}}\func{≡}\func{AL/}\eqRelR{\func{R}}$:

\ExecuteMetaData[agda/latex/Section5.tex]{ALUnionAssoc}
\end{example}

We can also see our queue example from \cref{subsec:queues} as an instance of \cref{thm:qer-descends}, as the graph of
$\func{appendReverse} : \func{BatchedQueue}\;A \to \func{ListQueue}\;A$ defines a QER. (The graph of any
function is zigzag-complete.) The quotient of $\func{BatchedQueue}\;A$ winds up being equivalent to our
hand-rolled $\func{BatchedQueueHIT}$, while the quotient of $\func{ListQueue}\;A$ is the original type. Once
we have obtained the raw quotients, we can check that they satisfy any axioms we like; this usually reduces
straightforwardly to showing that the axioms hold on the raw types up to $\eqRelL{R}$ or $\eqRelR{R}$.

\section{Related and future work}
\label{sec:related}

In this paper, we have shown how to combine univalence and higher inductive
types to obtain internal relational representation independence results in
dependent type theory, specifically \CubicalAgda{}. Univalence captures
internally the principle that isomorphic types are equal; the structure identity
principle lifts univalence to structured types and structured isomorphisms;
and set quotients (via higher inductive types) allow us to improve structured
relations to isomorphisms, in order to obtain representation independence
results through the SIP, and validate axioms that a concrete implementation may
otherwise fail to satisfy.

Our work lies at the intersection of formalized mathematics and programming
language theory; we will navigate the related work moving from the former to the
latter.

\paragraph{Proof reuse and transfer}

As we discuss in \cref{sec:intro}, an important problem in modern proof
assistants is the automatic transfer of programs and proofs across related
mathematical structures. In the context of \Coq{},
\citet{MagaudBertot02,Magaud03} considered transferring proofs between
isomorphic types by means of an external plugin, but their work did not handle
dependent types. \Citet{CohenDenesMortberg13} designed a more general framework
for \texttt{CoqEAL}, which supported program and data refinements for arbitrary
relations. Unlike our work, they support \emph{partial} quotients of types,
e.g., admitting \func{ℚ} as a direct refinement of
\tyProduct{\func{ℤ}}{\func{ℤ}}. We believe these are not particularly useful in
practice; indeed, at the time of writing, the \texttt{CoqEAL} library no longer
uses partial quotients. Their work was implemented using a parametricity plugin
and proof search using typeclass instances, but does not handle dependently
typed goals, and therefore could only transport programs and not proofs between
related structures.

\Citet{TabareauTanterSozeau18,TabareauTanterSozeau19} have recently improved
this work by combining parametricity and typeclass based proof search \`a la
\texttt{CoqEAL} with (axiomatic) univalence. This new \Coq{} framework, called
\emph{univalent parametricity}, enables the transport of programs and proofs
between related structures; however, unlike \texttt{CoqEAL}, its scope is limited to
isomorphic types due to its reliance on univalence. We emphasize that while
their work relies on an axiom for univalence, it nevertheless achieves code
reuse in \Coq{} without losing computational content, by means of a careful
setup in which the axiom does not interfere with computationally-relevant parts,
and by using typeclass instances to mimic computation rules similar to those of
Cubical Type Theory. However, their approach does not achieve all the niceties
of Cubical Type Theory, including constructive quotients and functional and
propositional extensionality. \Citet{RingerYazdani+19} automatically build
isomorphisms in \Coq{} between inductive types and transport programs and proofs
across these isomorphisms. This is accomplished using very similar techniques to
those of \citet{TabareauTanterSozeau18}, but written as an external plugin
instead of using typeclass instance search.

Another major contribution of \citet{TabareauTanterSozeau19} is a solution to
the \emph{anticipation problem}, which concerns inferring an interface \emph{a
posteriori} and transporting programs and proofs across this inferred interface.
This is especially useful in dependent type theory, where concrete types enjoy
many definitional equalities that abstract types do not. For example, if one
parameterizes proofs by an implementation of $\mathbb{N}$ with \func{+}, the
abstract \func{+} will not reduce on any input, whereas a concrete
implementation will reduce on \anum{0}. It is therefore very useful to allow users to
develop libraries using concrete types while still enjoying the benefits of
automated program and data refinements. In the setting of representation
independence, however, one usually considers programs that \emph{are}
parameterized by an interface, and that interface determines the notion of
structure-preservation, making the anticipation problem less central. But even a
partial solution to the anticipation problem in our setting would improve the
applicability and ease of use of our framework.

Similar frameworks exist for other proof assistants. The \texttt{Isabelle/HOL}
code generator uses refinements to make abstract code executable
\citep{HaftmannKrauss+13}; a similar refinement framework has been developed in
\Coq{} by \citet{DelawarePitClaudel+15} to synthesize abstract datatypes and
generate \texttt{OCaml} code. There is also an external tool for \texttt{Isabelle/HOL}
called \texttt{Autoref} which uses a parametricity-based translation to achieve
\texttt{CoqEAL}-style refinements \citep{Lammich13}.

\paragraph{The structure identity principle}

In the philosophy of mathematics, \emph{structuralism} is an informal principle
of representation independence: although mathematical objects may be realized in
\emph{ad hoc} ways, reasonable mathematical arguments implicitly access these
structures only through interfaces \citep{Benacerraf65,Awodey13}. As the work on
proof transfer illustrates, formalized mathematics must often rely on this
principle either implicitly or explicitly \citep{CaretteFarmerKohlhase14}.

Researchers in HoTT/UF noticed very early on that univalence captures a form of
structuralism, and have subsequently formalized quite a few versions of the SIP
in HoTT/UF. The first type-theoretic formalization and proof of the SIP for
algebraic structures was given by \citet{CoquandDanielsson13}. As discussed in
\cref{subsec:building}, their notions of structure are given by maps $S : \Type
\to \Type$ with an action on equivalences. That action induces a notion of
structured equivalence, which is then required to agree with paths.

Section 9.8 of the HoTT Book \citep{HoTT13} contains a version of the SIP for
structured categories; \citet{AhrensLumsdaine17} develop a variation of this SIP
using displayed categories. These principles are phrased in terms of structured
\emph{functions} between structured types, splitting the definition of a
structure into two components: what it means to equip a type with structure, and
what it means for a function to preserve that structure. The SIP then applies to
structured functions that are also equivalences. This work is limited to
structures on sets, as opposed to general types, due to the difficulty of
formulating higher categories in HoTT/UF, but \citet{AhrensNorth+20} have
recently developed a higher generalization.

Our SIP is derived from a version proposed by \citet{Escardo19}; like that of
\citet{CoquandDanielsson13}, it is phrased in terms of equivalences rather than
functions, but follows the categorical SIPs by allowing an arbitrary definition
of structured equivalences rather than deriving one from an action on
equivalences. Our \func{UnivalentStr} closely resembles Escard\'{o}'s
\emph{standard notion of structure}, except for being stated in terms of
dependent paths; the two conditions are equivalent.

Our univalent structured relations follow the pattern of the categorical SIPs, with relations substituted for
functions: a notion of structured relations is univalent when structured relations that are equivalences
correspond to paths. However, our suitability conditions requires that even relations that are merely
``equivalence-like'' (i.e., are QERs) have some relationship 
with paths. This
increased power comes with a more limited applicability, seen in the lack of a suitable notion of structured
relation for general function types.

\paragraph{Internalizing relational parametricity}

We have shown that univalence and HITs together recover a rich notion of
representation independence, a concept which in programming languages is often
obtained through parametricity theorems. We note that there are other ``free
theorems'' which cannot be proven in \CubicalAgda{}; a simple example is the
contractibility of $(X : \Type) \to X \to X$.

Just as univalence internalizes the fact that type theory respects isomorphism,
one can consider dependent type theories with \emph{internalized parametricity}
\citep{KrishnaswamiDreyer13, BernardyCoquandMoulin15, NuytsVezzosiDevriese17,
NuytsDevriese18, CavalloHarper20, AbelCockx+20}. Parametricity applies to
representation independence in greater generality than univalence: for example,
it implies that any computation parameterized by a \func{Queue} produces the same
results for $\func{List}\;A$ and $\func{AssocList}\;A$, while we only have such
a result for the quotients described in \cref{sec:relational}. Nevertheless, it
is worth noting that internal parametricity is not \emph{stronger} than
univalence, but rather incomparable: an action on relations does not imply an
action on equivalences.

From the standpoint of formalized mathematics, one drawback of internalized
parametricity is that it is more difficult to combine with classical principles
than univalence. For example, it is incompatible with excluded middle for
h-propositions \citep[Theorem 10]{CavalloHarper20}, unlike univalence
\citep{KapulkinLumsdaine20}. Furthermore, there are not currently any
large-scale proof assistants which implement type theories with internal
parametricity.

\paragraph{Programming in HoTT/UF}

We are not the first to consider applications of univalence and HITs to
programming languages. \Citet{AngiuliMorehouse+16} model \texttt{Darcs}-style
patch theories as HITs; \texttt{HoTTSQL} \citep{ChuWeitz+17} uses a univalent
universe to define the semantics of a query language and prove optimizations
correct; and \citet{BasoldGeuversVanDerWeide17} discuss HITs for types often
used in programming, including modular arithmetic, integers and finite sets.
These examples predate implementations of Cubical Type Theory, and would likely
be significantly easier to formalize in \CubicalAgda{}.

Countless variations of finite (multi)sets have been considered in all the main
proof assistants for HoTT/UF; we limit our comparisons to those closest to our
own. The finite sets of \citet{BasoldGeuversVanDerWeide17} have been further
studied by \citet{FruminGeuvers+18}, and are defined by encoding finite subsets
of $A$ as the free join-semilattice on $A$. If one drops idempotency of the
union operation, one obtains a HIT equivalent to
$\setQuotient{\func{List}\;A}{\eqRelL{\func{R}}}$. They also discuss a variation
called \emph{listed finite sets} which are almost exactly
$\setQuotient{\func{List}\;A}{\eqRelL{\func{R}}}$, except that they have a path
constructor equating lists with duplicate elements. This HIT and a \emph{listed
finite multisets} version have been formalized in \CubicalAgda{} by
\citet{ChoudhuryFiore19}; we have proven their type equivalent to
$\setQuotient{\func{List}\;A}{\eqRelL{\func{R}}}$ as well as to a direct HIT
version of $\setQuotient{\func{AssocList}\;A}{\eqRelR{\func{R}}}$.
\Citet{Gylterud20} defines finite multisets using set quotients of W-types of
sets. As an application, \citet{ForsbergXuGhani20} define in \CubicalAgda{} an
ordinal notation system using a variation of listed finite multisets, which is
then proven equivalent to two other definitions of ordinal notation systems, and
both programs and proofs are transported between the definitions. These
transports are done in an ad-hoc way that can be done more elegantly with the
SIP. Prior to HoTT/UF, \citet{Danielsson12} formalized finite multisets as a
setoid of lists modulo ``bag equivalence,'' which is exactly the relation
\func{R} that we use to quotient lists in \cref{ex:fmset}.

Batched queues maintain the invariant that the first list is only empty if the
queue is empty. One can encode this invariant in the type of the
$\con{Q⟨\_,\_⟩}$ constructor, so that \var{enqueue} and \var{dequeue} are
required to preserve the invariant. This variation of queues is equivalent to
\func{BatchedQueueHIT} and has been formalized independently in \CubicalAgda{}
by \citet{GjorupVindum19}.

\paragraph{Abstract types}

There are of course countless papers on abstract types and representation
independence in programming languages, but we limit our discussion to the few
most relevant. Most representation independence theorems apply to structured
heterogeneous relations \citep{Mitchell86}, much like the ones we consider. In
the context of System F, \citet{Robinson94} uses invariance under isomorphism to
derive a simple form of representation independence for abstract types.

Extended ML \citep{KahrsSannellaTarlecki97} is an extension of Standard ML whose
signatures can include axioms drawn from boolean expressions extended with
quantifiers and termination predicates; the technique of \emph{algebraic
specification} \citep{SannellaTarlecki87} is a method for establishing such axioms by successive
structure-preserving refinements, much like \texttt{CoqEAL}.
Miranda's \emph{data types with laws} extend datatype constructors with a
limited ability to maintain invariants \citep{Turner85}; as with the normal-form
representations of quotients discussed in \cref{eg:q}, these, unlike HITs, have
the drawback of incurring runtime cost to maintain normal forms.

\paragraph{Future work}

A natural next direction is to formalize representation independence for more
sophisticated data structures, such as self-balancing binary search trees. Our
methodology already applies to these and other examples, but proving the
correspondences suitable would be more complex. We can also extend the structure
tactics described in \cref{sec:sip,sec:relational}, particularly to support
inductive types beyond \func{Maybe}. Note, however, that these tactics apply to
the types that appear in \emph{interfaces}---not in representations---which are
usually quite limited.

Finally, we would like to investigate whether our techniques are practical in
\Coq{} using axiomatic univalence and HITs. Although we rely heavily on these
features computing automatically in \CubicalAgda{},
\citet{TabareauTanterSozeau19} overcome similar difficulties with the univalence
axiom by cleverly using typeclass instances to mimic computation rules.


\begin{acks}
  We thank Robert Harper and Jonathan Sterling for helpful conversations about
  this work.

  The material in this paper is based upon research supported by the
  \grantsponsor{1}{Air Force Office of Scientific
  Research}{https://www.wpafb.af.mil/afrl/afosr/} under MURI grants
  \grantnum{1}{FA9550-15-1-0053} and \grantnum{1}{FA9550-19-1-0216}, and the
  \grantsponsor{2}{Swedish Research Council
  (Vetenskapsr\r{a}det)}{https://www.vr.se/} under \grantnum{2}{Grant
  No.~2019-04545}.
  The views and conclusions contained in this document are those of the authors
  and should not be interpreted as representing the official policies, either
  expressed or implied, of any sponsoring institution, the U.S. government, or
  any other entity.
\end{acks}

\bibliography{refs}

%

\end{document}